\documentclass[a4paper,11pt]{article}

\usepackage{amsmath,amsfonts,amssymb,amsthm}
\usepackage{cases}

\usepackage{fullpage}
\usepackage{libertine}

\usepackage[numbers]{natbib}

\usepackage{algorithmicx}
\usepackage{algpseudocode}
\usepackage{url}
\usepackage[hidelinks]{hyperref}
\usepackage[utf8]{inputenc}
\usepackage{graphicx}
\usepackage{booktabs}
\usepackage{algorithm}

\newtheorem{theorem}{Theorem}[section]

\newtheorem{lemma}[theorem]{Lemma}

\newtheorem{claim}[theorem]{Claim}

\theoremstyle{definition}
\newtheorem{definition}{Definition}
\newtheorem*{comment*}{Comment}

\usepackage{xcolor}
\usepackage{color,soul}
\usepackage{comment}
\usepackage{multirow}
\usepackage{epsfig, graphicx,fleqn,xspace,color,graphics}

\newcommand{\games}{diversity-seeking jump games~}
\newcommand{\game}{diversity-seeking jump game~}

\newcommand{\calT}{\mathcal{T}}
\newcommand{\calA}{\mathcal{A}}
\newcommand{\calB}{\mathcal{B}}
\newcommand{\calC}{\mathcal{C}}
\newcommand{\SW}{\text{SW}}
\newcommand{\OPT}{\text{OPT}}
\newcommand{\NE}{\text{NE}}
\newcommand{\PoA}{\text{PoA}}
\newcommand{\PoS}{\text{PoS}}

\usepackage{authblk}
\title{\bf Diversity-Seeking Jump Games in Networks\thanks{A preliminary version of this paper was published as \citep{narayanan2023diversity}. This version includes the complete proofs of all results, as well as a new set of results on the price of anarchy and price of stability (Section~\ref{sec:poa}).}}

\author[1]{Lata Narayanan\thanks{lata.narayanan@concordia.ca}}
\author[1]{Yasaman Sabbagh\thanks{yasaman.sabbaghziarani@concordia.ca}}
\author[2]{Alexandros A. Voudouris\thanks{alexandros.voudouris@essex.ac.uk}}

\affil[1]{Department of Computer Science and Software Engineering, Concordia University}
\affil[2]{School of Computer Science and Electronic Engineering, University of Essex}

\date{}

\begin{document}

\maketitle

\begin{abstract}
Recently, strategic games inspired by Schelling's influential model of residential segregation have been studied in the TCS and AI literature. In these games, agents of $k$ different types occupy the nodes of a network topology aiming to maximize their utility, which is a function of the fraction of same-type agents they are adjacent to in the network. As such, the agents exhibit similarity-seeking strategic behavior. 
In this paper, we introduce a class of strategic jump games in which the agents are {\em diversity-seeking}: The utility of an agent is defined as the fraction of its neighbors that are of {\em different} type than itself. We show that in general it is computationally hard to determine the existence of an equilibrium in such games. However, when the network is a tree, diversity-seeking jump games always admit an equilibrium assignment. For regular graphs and spider graphs with a single empty node, we prove a stronger result: The game is potential, that is, the improving response dynamics always converge to an equilibrium from any initial placement of the agents. We also show (nearly tight) bounds on the price of anarchy and price of stability in terms of the social welfare (the total utility of the agents). 
\end{abstract}

\section{Introduction}
In his seminal work \citep{schelling1969models,schelling1971dynamic}, the economist Thomas Schelling proposed an elegant model to explain the phenomenon of residential segregation in American cities. In his model, agents of two types are placed uniformly at random on the nodes of a network. An agent is happy if at least some fraction $\tau$ of its neighbors are of the same type as itself, and unhappy otherwise. An unhappy agent will {\em jump} to an unoccupied node in the network, or exchange positions with another unhappy agent of a different type. Schelling showed experimentally that this random process would lead to {\em segregated neighborhoods}, even for $\tau \approx 1/3$. His work showed that even small and local individual preferences for one's own type can lead to large-scale and global phenomena such as residential segregation. As such, it inspired a significant number of follow-up studies in sociology and economics \citep{benenson2009schelling,pancs2007schelling,rogers2011unified}, with many empirical studies being conducted to study the influence of different parameters on segregation models \citep{benard2007wealth,bruch2014population,bursell2018diversity,fossett1998simseg,stivala2016diversity}. 

Inspired by Schelling's model, there have been a few lines of inquiry pursued by computer scientists as well. 
Some researchers attempted to prove analytically the conditions under which segregation takes place \citep{brandt2012analysis,immorlica2017exponential}, while a more recent line of research used {\em game-theoretic} tools to study the problem \citep{ijcai2022p22,bullinger2021welfare,chauhan2018schelling,echzell2019convergence,ijcai2019p38,kanellopoulos2020modified}. In the latter setting, strategic agents are placed at arbitrary positions in the network and move to new positions to improve their {\em utility}, defined as a function of the fraction of their neighbors which are of the same type as themselves. Two main classes of games have been considered. In {\em jump} games, the agents can move to previously unoccupied locations aiming to improve their utility, while in {\em swap} games, unhappy agents of different types can swap their locations if this increases the utility of both of them. Researchers have mainly studied questions related to the {\em existence}, {\em quality}, and {\em computational complexity} of {\em equilibria}, which are stable assignments of agents to the nodes of the network such that no agent wants or can deviate.

\subsection{Our Contribution}
In this paper, we introduce a new class of strategic games, called {\em diversity-seeking jump games}. In such games, the utility of an agent is the fraction of its neighbors of a {\em different} type than itself. To the best of our knowledge, this is the first time such a utility function has been studied in the context of jump games on networks that are inspired by Schelling's model. Data from the General Social Survey \cite{smith2019general} show that people prefer diverse neighborhoods rather than segregated ones. This survey, which is regularly conducted in the US since 1950, shows that the percentage of people preferring diverse neighborhoods has been steadily increasing. There are also other settings that could be modeled using this new utility function; e.g., teams composed of people with diverse backgrounds and skill sets to bring a broader range of perspectives to their business, or research groups composed of people bringing expertise from different disciplines. Indeed, many studies show that ethnically and gender-wise diverse teams lead to better outcomes for business \citep{antonio2004effects,dezso2012does,elia2019impact,filbeck2017does,freeman2015collaborating,loyd2013social,parker2016impact,phillips2006surface,sommers2006racial}.

It may seem at first glance that Schelling jump games are intimately related to our diversity-seeking jump games, but the relationship is not obvious. Observe that the utility of an agent in a Schelling {\em similarity-seeking} jump game is not always equal to the complement of its utility in the diversity-seeking setting. To see this, consider an agent placed at a node of degree $1$ that is adjacent to an empty node.  Such an isolated agent has utility $0$ in both similarity-seeking and diversity-seeking settings, thus showing that the models are not simply complements of each other. 

We present a series of results related to the computational complexity, the existence, and quality of equilibrium assignments in diversity-seeking jump games. We first show that given a network topology and a number of agents drawn from a set of types, it is NP-hard to determine whether there exists an equilibrium if some of the agents are {\em stubborn} (and do not ever move from the nodes they are initially placed). This hardness reduction can be found in Section~\ref{sec:hardness}.

In terms of positive existence results, in Sections~\ref{sec:tree} and \ref{sec:bounded}, we show that there are always equilibrium assignments in games where the network topology is a tree graph with any number of empty nodes, if all agents are strategic. 
For some classes of graphs, we show a stronger result: The \game is a {\em potential game} when the topology is a regular graph or a spider graph with a single empty node, or a line with any number of empty nodes. This means that from, any initial placement of agents to the nodes of the topology, improving response dynamics (IRD) by the agents always converges to an equilibrium. 
Our positive results show a sharp difference between \games and similarity-seeking jump games, for which there are instances with tree topologoes that do not admit equilibria \cite{ijcai2019p38}, and instances with spider graph topologies which are not potential games. 

Finally, we show that there is not much loss of efficiency at equilibrium assignments in diversity-seeking jump games in terms of social welfare, defined as the total utility of all agents. In particular, we show that the price of anarchy (that is, the worst-case ratio of the optimal social welfare over the minimum social welfare achieved at equilibrium) is at most $n/(k-1)$ when there are $n$ strategic agents partitioned into $k \geq 2$ types, and at most $k/(k-1)$ when the agents are partitioned into $k$ {\em symmetric} types consisting of $n/k$ agents each. These bounds are (nearly) tight even for simple topologies, such as a line graph in some cases or a star graph in other cases. We remark that, in contrast to similarity-seeking games where the price of anarchy has been shown to be increasing in $k$ for symmetric types~\cite{ijcai2019p38,kanellopoulos2023tolerance}, it is decreasing in $k$ for 
diversity-seeking games and tends to $1$. Intuitively, this is because in the games we study here, the social welfare is a measure of integration which is easier to achieve when there are a lot of different types of agents, whereas the social welfare is a measure of segregation in the model of \citet{ijcai2019p38} which is harder to achieve when there are many types. 
We also prove a lower bound of $65/62 \approx 1.048$ on the price of stability (which is the worst-case ratio of the optimal social welfare over the maximum social welfare achieved at equilibrium), thus showing that optimal assignments are not necessarily equilibria. These results can be found in Section~\ref{sec:poa}.

\subsection{Related Work}
In the last decade, there have been attempts in the computer science literature \citep{brandt2012analysis,immorlica2017exponential} to theoretically analyze the random process described by Schelling. It was shown that the expected size of the resulting segregated neighborhoods is polynomial in the size of the neighborhood on a line topology~\citep{brandt2012analysis} and exponential in its size on a grid topology~\citep{immorlica2017exponential}; however, in both cases it is independent of the overall number of agents. 

The study of the above process as a strategic game played by two types of agents rather than as random process appears to have been initiated by \citet{zhang2004dynamic}, who introduced a game-theoretic model where agents have a noisy single-peaked utility function that depends on the ratio of the numbers of agents of two agent types in any local neighborhood. \citet{chauhan2018schelling} introduced a game-theoretic model that incorporates Schelling's notion of a tolerance threshold. In their model, each agent has a threshold parameter $\tau \in (0,1)$ and may have a preferred location. The primary goal of an agent is to find a location where the happiness ratio exceeds $\tau$; if such a location does not exist, the agent aims to maximize its happiness ratio. Its secondary goal is to minimize the distance to its preferred location. \citet{chauhan2018schelling} studied the convergence of best-response dynamics to an equilibrium assignment in both jump and swap games with two types of agents and for various values of the threshold parameter. \citet{echzell2019convergence} generalized the model of \cite{chauhan2018schelling} by considering games with agents of $k \geq 2$ different types, in which the cost of an agent is related to the ratio of agents of its own type over the whole neighborhood, or the ratio of agents of its own type over the part of the neighborhood that includes the agents of the majority type different than its own. They showed results related to the converge of the dynamics and the computational hardness of finding placements that maximize the number of happy agents. 

More related to our model, \citet{ijcai2019p38} studied similarity-seeking jump games (that is, $\tau=1$) for $k \geq 2$ types of agents that can be either  strategic (aiming to maximize their utility) or stubborn (who stay at their initial location regardless of the composition of the neighborhood). \citeauthor{ijcai2019p38} showed that, while equilibria always exist when the topology is a star or a graph of degree $2$, they may not exist even for trees of degree $4$. They also showed that it is computationally hard to find equilibrium assignments or assignments with optimal social welfare, and bounded the price of anarchy and stability for both general and restricted games. \citet{agarwal2020swap} considered the utility function of \cite{ijcai2019p38}, but focused on swap games and, besides the social welfare, also considered a new social objective called the {\em degree of integration}. They showed similar results about the existence and the computational complexity of equilibrium assignments. An extended version of these two works that includes more results was later published as \citep{schelling-journal}. The hardness results of \citep{schelling-journal} where improved by \citet{kreisel2021equilibria}, who showed that the problem of determining the existence of equilibria is NP-hard in both jump and swap games, even if all agents are strategic. \citet{bilo2022topological} investigated the influence of the underlying topology and locality in swap games when agents can only swap with other agents at most a number of hops away, and also improved some price of anarchy results of \citep{schelling-journal}. 

Other variants of the model have been considered by \citet{kanellopoulos2020modified} who focuses on a slightly different utility function where an agent is also counted as part of its neighborhood, by \citet{chan2020schelling} who considered social Schelling games where the agents are not partitioned into type but rather are connected via a social network, and by \citet{bilo2023continuous} who considered a model in which the types are continuous rather than discrete (as in most papers). The computational complexity of finding assignments with high objective value according to different functions, including the social welfare, the Nash welfare, and other optimality notions has been considered by \citet{bullinger2021welfare}, and also by \citet{deligkas2024parameterized} from the parameterized complexity perspective.  

Some other recent papers that are related to ours are that of \citet{ijcai2022p12}, \citet{friedrich2023single} and \citet{kanellopoulos2023tolerance}, which are also motivated by the observation that real-world agents might favor diverse neighborhoods to some extend. \citet{ijcai2022p12} and \citet{friedrich2023single} considered swap and jump games, respectively, in which agents of two types have single-peaked utility functions that increase monotonically with the fraction of same-type neighbors in the interval $[0, \Lambda]$ for some $\Lambda \in (0,1)$, and then decrease monotonically afterward; this is in contrast to our model here, where the utility is a monotone function. \citet{kanellopoulos2023tolerance} considered the case where there are $k \geq 2$ types of agents and there is an implicit ordering of the types such that the distance between the types according to this ordering determines how much different types affect the utility of an agent whenever there are agents of those types in its neighborhood. All these papers showed instances where equilibria do not always exist, but also identified restricted classes of games for which existence of equilibria is guaranteed either by explicit constructive algorithms, or because IRD converges. They also present tight bounds on the price of anarchy and the price of stability in terms of the social welfare or the degree of integration. 

\section{The Model and Notation}
In an instance $I = (R, S, \calT, G, \lambda)$ of the {\em \game}, there is a set $R$ of {\em strategic agents} and a set $S$ of {\em stubborn agents} with $|R \cup S| = n \geq 2$. The agents are further partitioned into  $k \geq 2$ different {\em types} $\calT = (T_{1}, \ldots , T_{k})$. Each agent occupies a node of an undirected graph $G=(V,E)$ that satisfies $|V| > n$; $G$ is referred to as the {\em topology}. The locations (nodes of $G$) of the stubborn agents are fixed and given by an injective mapping $\lambda$ from $S$ to $V$. In contrast, the strategic agents can choose their locations freely, subject to the constraint that no two agents can share the same node. An {\em assignment} $C$ determines the node $v_\calA(C)$ that each agent $\calA$ (strategic or stubborn) occupies on the topology, that is, it is an injective mapping of agents to the nodes that respects $\lambda$; let $\calC(I)$ be the set of all possible assignments for game $I$. 

We define $n_T(v, C)$ to be the number of neighbors of node $v$ that are occupied by agents of type $T \in \cal{T}$ according to assignment $C$. Let $n(v,C) = \sum_{T \in \calT}n_{T}(v, C)$ be the total number of agents in $v$'s neighborhood according to assignment $C$. The satisfaction of the agents for an assignment $C$ is measured by a {\em utility function}. 
In particular, the utility of an agent $\calA$ of type $T$ that occupies node $v_\calA(C)$ in assignment $C$ is defined as the fraction of agents of type different than $T$ in $\calA$'s neighborhood. That is, 
$$u_{\calA}(C) = \frac{\sum_{t \neq T} n_t(v_{\calA}(C),C)}{n(v_{\calA}(C),C)}$$
Observe that $u_{\calA}(C) \in [0,1]$. 
By convention, the utility of an agent without neighbors is $0$.


\subsection{Equilibria and Improving Response Dynamics}
We are interested in stable assignments of diversity-seeking jump games, in which strategic agents do not have incentive to deviate by jumping to empty nodes of the topology to increase their utility. 

\begin{definition} \label{def:equilibrium}
An assignment $C$ is said to be a {\em pure Nash equilibrium} (or, simply, {\em equilibrium}; NE) if $u_\calA(C) \geq u_\calA(C')$ for every strategic agent $\calA$ and any assignment $C'$ according to which $\calA$ jumps to a node that is empty in $C$ while all other agents occupy the same nodes as in $C$. We will denote by $\NE(I)$ the set of all equilibrium assignments of game $I$. 
\end{definition}

More specifically, we study the convergence of {\em improving response dynamics (IRD)} to equilibrium assignments in diversity-seeking jump games: 
In each step, a single strategic agent tries to move to an unoccupied node where it gets higher utility. We say such a move of an agent $\calA$ from assignment $C$ to assignment $C'$ is an {\em improving move} if $u_{\calA}(C) < u_{\calA}(C')$. Note that in assignment $C$, a single agent may have many different improving moves, and multiple agents may have improving moves. We do not make any assumptions about which improving move will take place if multiple ones are available. Clearly, the IRD converges to an equilibrium assignment if no agent has an improving move. 

A game is a {\em potential game} if and only if there exists {\em a generalized ordinal potential function}, that is, a non-negative real-valued function $\Phi$ on the set of assignments such that for any two assignments $C$ and $C'$ where there is an improving move from $C$ to $C'$, we have $\Phi(C') < \Phi(C)$. As observed by \citet{monderer1996potential}, such games always admit an equilibrium, and furthermore, regardless of the starting assignment, any sequence of improving moves is finite and will terminate in an equilibrium. 

Not all games are potential ones in case there is an {\em improving response cycle (IRC)} is a sequence of improving moves that leads to a repeating assignment. In particular, the existence of an IRC implies that the IRD does not always converge to an equilibrium, but it does not imply the non-existence of an equilibrium. 

For particular topologies that we will consider in later sections, we will show that the following function is a generalized ordinal potential function. Similar potential functions have been used before in different models~\citep{chauhan2018schelling,echzell2019convergence,ijcai2019p38}. 

\begin{definition} \label{def:potential-function}
For an assignment $C$ and some value $m \in (0,1)$, let  $\Phi(C) = \sum_e{w_{C}(e)}$, where the weight $w_{C}(e)$ of any edge $e = \{u,v\} \in E$ is defined as:
\begin{align*}
w_{C}(e) = 
\begin{cases}
1   & \text{if $u$ and $v$ are occupied by agents of the same type}\\
m &\text{if either $u$ or $v$ is unoccupied}\\
0&\text{otherwise}.
\end{cases}
\end{align*}
\end{definition}

The following lemma shows that, if a diversity-seeking jump game without stubborn agents is a potential game for a particular class of topologies, then it is also a potential game if some of the agents are stubborn. This will allow us to consider only the case when all agents are strategic when showing that a game is a potential one for a class of networks.

\begin{lemma}\label{lem:stubborn}
For any given $k\geq2 $ and topology $G$, if the game $I = (R, \varnothing, T, G, \lambda)$ is a potential game, then every game $I' = (R', S', T, G, \lambda)$ in which $S'$ and $R'$ form a partition of $R$, is also a potential game. 
\end{lemma}
\begin{proof}
Suppose $I = (R, \varnothing, T, G, \lambda)$ is a potential game for a given topology $G$, and let $\cal{C}$ be the set of valid assignments for $I$. Now consider the game $I' = (R', S', T, G, \lambda)$ where $S' \neq \varnothing$, $R' \subset R$ and $S' \subset R$, and let $\cal{C'}$ be the set of valid assignments for $I'$. Since $R' 
\subset R$, clearly $\cal{C} \subset {C'}$. Also any move between two assignments in $I'$ is also possible in $I$, though the converse does not  hold. It follows that if there is no IRC in $I$, there cannot be one in $I'$. 
\end{proof}

\subsection{Price of Anarchy and Stability}
We are also interested in the welfare guarantees of equilibrium assignments in terms of the {\em social welfare} that is defined as the total utility of all strategic agents. In particular, for an assignment $C$, the social welfare is 
\begin{align*}
    \SW(C) = \sum_{\calA \in R} u_\calA(C).
\end{align*}
For a given game $I$, let $\OPT(I) = \max_{C \in \calC(I)} \SW(C)$ be the optimal social welfare that can be achieved among all possible assignments of $I$. The {\em price of anarchy} is a pessimistic measure of the loss of welfare at equilibrium and is defined as the worst-case ratio, over all possible games $I$ with $\NE(I) \neq \varnothing$, of the optimal social welfare $\OPT(I)$ over the {\em minimum} social welfare achieved among all equilibrium assignments, that is, 
\begin{align*}
    \PoA = \sup_{I: \NE(I) \neq \varnothing} \frac{\OPT(I)}{\min_{C \in \NE(I)}\SW(C)}.
\end{align*}
On the other hand, the {\em price of stability} is an optimistic measure of the loss of welfare at equilibrium and is defined as the worst-case ratio of the optimal social welfare over the {\em maximum} social welfare at equilibrium, that is, 
\begin{align*}
    \PoS = \sup_{I: \NE(I) \neq \varnothing} \frac{\OPT(I)}{\max_{C \in \NE(I)}\SW(C)}.
\end{align*}
Clearly, the price of anarchy and the price of stability are both always at least $1$; the closer they are to $1$, the smaller the loss of welfare at equilibrium. Besides showing bounds on these two measures for all possible games, we will also focus on specific classes of games with particular topologies. 

\section{Complexity of Computing Equilibria} \label{sec:hardness}
In this section, we study the computational complexity of determining if a given diversity-seeking jump game admits an equilibrium assignment.

\begin{theorem}\label{thm:hardness}
For any $k\geq2 $, it is NP-complete to decide whether a given a \game $I = (R, S, T, G, \lambda)$ admits an equilibrium assignment.
\end{theorem}

\begin{proof}
We give a proof for $k=2$; it is straightforward to extend it to $k \geq 2$. We can clearly verify if a given assignment is an equilibrium, so the decision problem belongs to NP. To show hardness, we will use a reduction from the Independent Set (IS) decision problem \cite{garey1979computers}. Recall that an independent set in a graph is a subset of vertices of the graph such that no two vertices in the subset are connected by an edge. An instance of the IS problem is an undirected graph $H= (X,Y)$, where $X$ and $Y$ are the set of vertices and edges respectively, and an integer $s$; it is a yes-instance if and only if $H$ has an independent set of size $ \geq s$. We construct a \game $I$ as follows (see also Figure~\ref{fig:hardness}): 
\begin{enumerate}
    \item There are two agent types: red and blue.
    \item There are a total of $8s+16$ agents; $s+1$ of them are strategic red agents, $2s+4$ are stubborn red agents, and $5s+11$ are stubborn blue agents. 
    \item The topology $G = (V,E)$ consists of two components $G_1$ and $G_2$, as follows:
    \begin{itemize}
        \item $G_1 = (V_1, E_1)$, where $V_1 = X \cup W$, $|W| = 7s+1$, and $E_1 = Y \cup \{\{v,w\}: v \in X, w \in W\} $; $5s+1$ stubborn blue agents and $2s$ stubborn red agents are placed at the nodes of $W$.
        \item $G_2$ has exactly three nodes (denoted $x$, $y$, and $z$) available to strategic agents, while all the other nodes are occupied by stubborn agents. 
        Node $x$ is connected to $y$ and one blue agent. 
        Node $y$ is connected to six nodes, containing one red agent and five blue agents. 
        Finally, $z$ is connected to 7 nodes containing 4 blue agents and 3 red agents. 
        Observe that when the other two nodes in $G_2$ are unoccupied, $x, y$, and $z$ offer utilities $1$, $5/6$ and $4/7$, respectively. 
        \item There is an edge between two (arbitrarily chosen)  nodes containing stubborn agents in $G_1$ and $G_2$, thereby connecting the graph. 
    \end{itemize}
\end{enumerate}

\begin{figure}[t]
    \centering
    \includegraphics[scale=0.45]{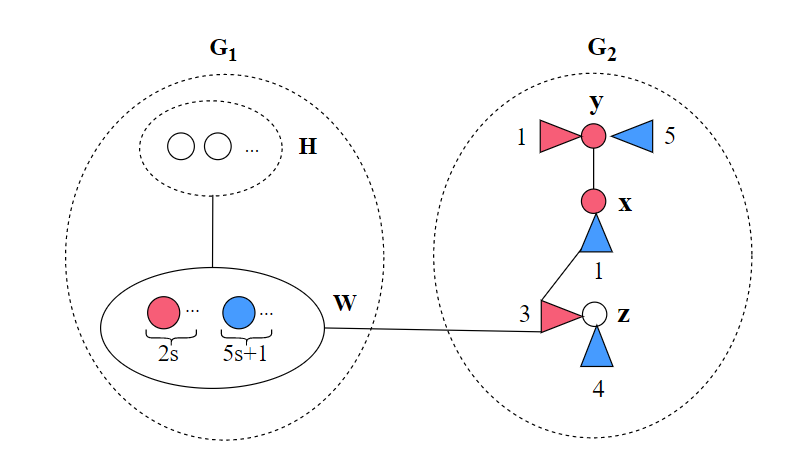}
    \caption{$G_1$ and $G_2$ are the two components of $G$, connected by a single edge. Every node in $H$ is connected to every node in $W$. The blue (red) triangle connected to $x$ denotes a set of neighbors of $x$  with blue (red) agents, and the number next to a triangle denotes the number of nodes in the corresponding set.}
    \label{fig:hardness}
\end{figure}
    
The main idea of the reduction is as follows: We will show that, if an IS of size $ \geq s$ exists in $H$, then placing $s$ strategic agents at the nodes of the IS and the remaining strategic agent at node $x$ is an equilibrium assignment in $I$. Conversely, if no IS of size at least $s$ exists in $H$, then we will show with a case analysis that there cannot be an equilibrium assignment in $I$.

We proceed with the proof that $I$ admits an equilibrium if and only if $H$ has an IS of size at least $s$. 
First, suppose that $H$ contains an IS of size at least $s$. 
Consider the following assignment in $I$: $s$ red agents occupy the nodes of the IS in $X$ and one red agent occupies node $x$ in $G_2$. 
We argue this is an equilibrium. The utility of each agent at a nodes the IS is $\frac{5s+1}{7s+1}$. Note that any remaining empty node in $G_1$ offers a utility that is at most $\frac{5s+1}{7s+1} < 1$, while the utility of the agent at $x$ is $1$. 
Therefore, the agent at $x$ does not benefit by moving to an empty node in $G_1$. Also, note that when node $x$ is occupied, the highest utility a red agent can achieve in $G_2$ is $5/7$ by moving to $y$. We have $\frac{5s+1}{7s+1} > 5/7$, therefore no agent in $G_1$ has an improving move to an empty node in $G_2$, and the assignment is an equilibrium.
    
On the other hand, suppose that $H$ does not contain an IS of size at least $s$. Suppose towards a contradiction that $I$ admits an equilibrium assignment. There are four possible cases: 
\begin{enumerate}
    \item All available nodes in $G_2$ are empty. 
    Then since there is no IS of size $s$ in $H$, and therefore in $G_1$, at least one strategic red agent is adjacent to another strategic red agent and has utility at most $\frac{5s+1}{7s+2}$. Since $\frac{5s+1}{7s+2}  < 1$, that agent can increase its utility by moving to node $x$ in $G_2$. Therefore, this cannot be an equilibrium assignment.
    
    \item Exactly one available node in $G_2$ is occupied by a strategic agent. 
    Then this agent must occupy $x$ as it provides the highest utility. As in the previous case, since there is no IS of size $s$, at least one strategic red agent is adjacent to another strategic red agent in $G_1$ and has utility at most $\frac{5s+1}{7s+2}$. Since $\frac{5s+1}{7s+2}  < 5/7$, that agent can improve its utility by moving to node $y$ in $G_2$. Therefore, this also cannot be an equilibrium assignment.
    
    \item Exactly two available nodes in $G_2$ are occupied. 
    If $x$ and $y$ are occupied, the agent at $x$ has utility $1/2$ and is motivated to move to $z$ to get utility $4/7$. 
    If $x$ and $z$ are occupied, the agent at $z$ has utility $4/7$ and is motivated to move to $y$ to get utility $5/7$. 
    If $y$ and $z$ are occupied, the agent at $y$ has utility $5/6$ and is motivated to move to $x$ to get utility $1$. 
    In all cases, there is an agent that wants to move to increase its utility, so this cannot be an equilibrium assignment.
    
    \item All available nodes in $G_2$ are occupied. 
    Every empty node in $G_1$ offers a utility that is at least $\frac{5s+1}{(7s+1) + (s-2)} = \frac{5s+1}{8s-1}$. Note that the agent at $x$ has utility $1/2 < \frac{5s+1}{8s-1}$. Therefore, the agent at $x$ is motivated to move to an empty node in $G_1$. Therefore, this cannot be an equilibrium assignment. 
\end{enumerate}
We have shown by contradiction that, if $H$ does not contain an IS of size $s$, there is no equilibrium assignment in $I$. 
This completes the proof of the NP-hardness. 
\end{proof}

The proof of the above hardness resilt above relies heavily on the existence of stubborn agents, as in the proofs of similar results for similarity-seeking jump games in \cite{ijcai2019p38}. 
In a recent paper, \citet{kreisel2021equilibria} showed that deciding the existence of an equilibrium in a similarity-seeking game is NP-complete, even when all agents are strategic. Their techniques are not immediately applicable to our setting, and in fact we have not been able to construct a diversity-seeking jump game with only strategic agents that does not have an equilibrium.

\section{Seeking Diversity in Tree Topologies} \label{sec:tree}
In this section, we consider \games with tree topologies. 
We first show that IRD do not always converge such games, even with a single empty node. 
However, there is always an equilibrium assignment; we give a polynomial-time algorithm to find one. 
In contrast, recall that similarity-seeking games do not always admit an equilibrium when the topology is a tree \citep{ijcai2019p38}. In Section~\ref{subsec:Spider-graphs}, we will show that the game is a potential game when the topology is a spider graph (i.e., a tree in which there is a single node of degree at least $3$).  

\begin{figure}[!htb]
    \centering
    \includegraphics[scale=0.38]{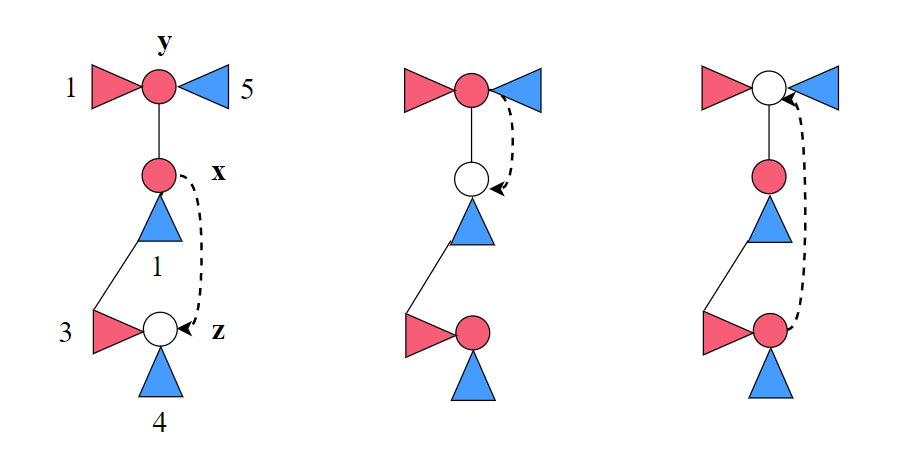}
    \caption{An IRC in a game with a tree topology. The blue (red) triangle connected to $x$ denotes a set of neighbors of $x$ with blue (red) agents, and the number next to a triangle denotes the number of nodes in the corresponding set.}
    \label{fig:cycleTree}
\end{figure}

\begin{theorem} \label{thm:tree-not-pot}
For every $k \geq 2$, there exists a \game $I = (R, S, T, G, \lambda)$ such that $G$ is a tree and $I$ is not a potential game. 
\end{theorem}

\begin{proof}
Figure~\ref{fig:cycleTree} shows an IRC in a diversity-seeking jump game on a tree with a single empty node and two types of agents. It is easy to verify that each move shown is improving. 
\end{proof}

In fact, if the agents at the nodes of the triangles in Figure~\ref{fig:cycleTree} are all stubborn and only the agents that occupy $x$, $y$ or $z$ are strategic, then this \game does not admit an equilibrium. Thus, equilibria are not guaranteed to exist when the topology is a tree and there exist stubborn agents. Nevertheless, we show that there is always an equilibrium assignment when the topology is a tree if all agents are strategic; we give an algorithm to find such an assignment. As a warm-up, we first consider the case where the tree has a single empty node. 

\begin{theorem}\label{thm:tree-empty-node}
Given $k\geq2$, every \game $(R, S, T, G, \lambda)$ where $G$ is a tree, $|V| = |R| + 1 $ and $S = \varnothing$ admits an equilibrium assignment. 
\end{theorem}

\begin{proof}
Pick a node $r$ of degree one to be the root of the tree and call its unique neighbor as $v$. Assume the agents are pre-ordered by type. The algorithm proceeds in two phases. In the first phase, starting with the deepest {\em odd} level in the tree, place agents in the current level $\ell$, moving from left to right. When the current level is completely filled, skip a level, go to level $\ell-2$ and repeat, until the unique node $v$ in level $1$ is filled, say by an agent of a type $T_x$: For convenience, we call it a {\em red agent}. This ends the first phase. In the second phase, move downwards from level $2$, filling even levels from left to right.  

See Figure~\ref{fig:treeSingle} for an example. At the end of this procedure, only the root node is empty. Furthermore, all non-red agents are guaranteed to have utility one, and have no incentive to move. There may be red agents that have utility less than $1$, but they cannot improve their utility by moving to the root. This proves that the assignment is an equilibrium. 
\end{proof}

\begin{figure}[t]
    \centering
    \includegraphics[scale=0.38]{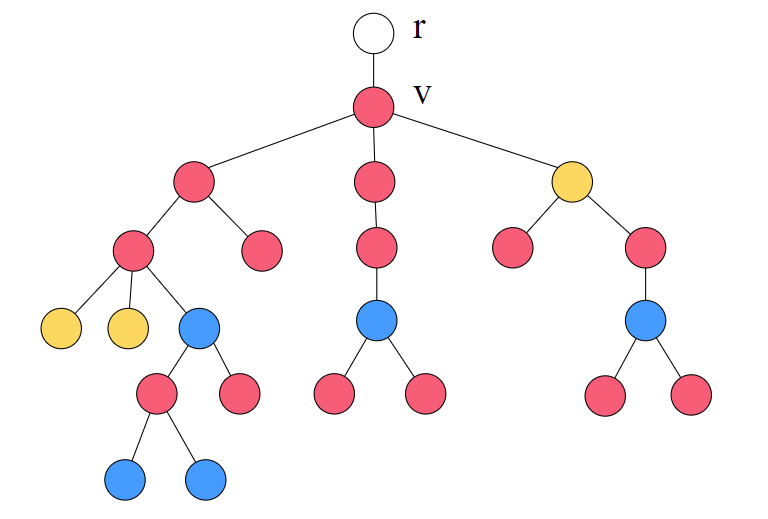}
    \caption{An equilibrium assignment for an instance with $k=3$. The red agents are placed first, starting at level 5, and ending at level 2, followed by yellow, and blue agents respectively. By convention, $r$ is on level 0.}
    \label{fig:treeSingle}
    \end{figure}

For the case where the number of empty nodes is more than $1$, we first find a connected sub-tree with one more node than the number of agents, and find an equilibrium assignment for the sub-tree using the method of Theorem~\ref{thm:tree-empty-node}. Next we adjust this assignment to get an assignment for the original tree by moving some agents to empty nodes in the original tree that were not present in the sub-tree; this ensures that no agent will have an improving move. Before we provide the full proof, we present some few easily seen properties of the assignment for trees described above:
\begin{itemize}
\item[(P1)] Every non-red agent is either placed on odd levels or on even levels but not both. This implies that for every non-red agent $\calA$ placed on  a particular level, there are no agents of its own type in the previous or next levels. 

\item [(P2)] Except possibly for two red agents, that we call {\em mixed} agents, a red agent has either all red children\footnote{For brevity, we say an agent is a child (neighbor/parent) of another agent to mean an agent is placed at a node that is the child (neighbor/parent) of the node where the other agent is placed.} or all non-red children. This is because there are only two levels that can have both red and non-red agents, namely the levels where we start and stop placing red agents.  

\item [(P3)] Every red agent with a non-red parent can only have non-red children, and therefore has utility one. To see this, note that the parent and the children of this red agent are either both on odd levels or both on even levels. Additionally, note that red agents are the last agents to fill the odd levels and the odd levels are filled from bottom to top so if an odd level does not contain all red agents, no odd level below that level contain red agents. Similarly, even levels are filled from top to bottom, with red agents being the first to fill them, so if an even level does not contain all red agents, no even level below that level contain red agents. 
\end{itemize}

We are now ready to show that there is always an equilibrium when the topology is a tree. 

\begin{theorem}\label{thm:tree-many-empty-nodes}
Given $k\geq2$, every \game $(R, S, T, G, \lambda)$ where $G$ is a tree, $|V| > |R| + 1 $ and $S = \varnothing$ admits an equilibrium assignment. 
\end{theorem}

\begin{proof} 
Given the tree $G = (E,V)$, fix a root node of degree one, and repeatedly remove leaf nodes until we have a tree with exactly $|R|+1$ nodes, call it $G' = (E',V')$. Let $C_{G'}$, denote the equilibrium assignment in $G'$ as described in the proof of Theorem~\ref{thm:tree-empty-node}. As before, let us call the type of the agent assigned to the unique neighbor of the root a red agent. 

Now consider the original tree $G$, and consider the same placement of agents as in $C_{G'}$, call this assignment $C_G$. Clearly the utility of agents in $C_G$ is exactly the same as their utility in $C_{G'}$, as acquiring new empty neighbors does not change an agent's utility. From Property (P1), all non-red agents have utility $1$. As in $C_{G'}$, there may be some red agents that have utility less than one in $C_G$, but in $G$, they may have improving moves available, i.e., $C_G$ may not be an equilibrium.

Next we show how to convert $C_G$ into an equilibrium assignment, by changing the locations of some agents one by one in such a way that the agent being moved increases its utility and no other agent decreases its utility. We first do some preprocessing with the children of mixed agents. By property (P2), there are at most 2 red mixed agents. Let $\calA$ be such an agent; if there are two, we call the agent closer to the root $\calA$ and the one further from the root $\calB$. Notice  that $\calA$ and $\calB$ must be on adjacent levels. Suppose not, then notice that the level below $\calA$ and the level on which $\calB$ is located are both odd or both even. Then $\calB$ is either on an odd level below the level we start placing red agents or on an even level below the level we finish placing the red agents. In both cases, $\calB$ cannot contain red agents, therefore $\calA$ and $\calB$ must be on adjacent levels.
 
For the preprocessing, we make a local adjustment to the placement of children of $\calA$. Let $S_1$ be the set of red children of $\calA$ that are adjacent to at least one more agent, and let $S_2$ be its  non-red children with {\em no} other agents as neighbors. If $S_2 = \varnothing$, and the second mixed agent $\calB$ exists and is a child of $\calA$, we swap\footnote{Note that this is a jump game, so the swap referred to here is not to be confused with improving response dynamics.} the position of $\calB$ with any non-red child of $\calA$. Note that the non-red child in its new position still has utility $1$, as do all non-red children of $\calA$ in this case. Another consequence of this swap is that $\calB$ is no longer a mixed agent, so there is only one mixed agent in the tree. 

If $S_2 \neq \varnothing$, we swap the positions of $\min(|S_1|, |S_2|)$ agents in $S_2$ with the agents in the other set. If $\calB$ is a child of $\calA$, then we make sure that $\calB$ participates in the swap. 

\begin{claim}
After the preprocessing operation, {\em at least one} of the following two conditions holds: 
\begin{enumerate} 
    \item  All red children of $\calA$ have utility $0$.
    \item  All of $\calA$'s non-red children have children of their own. 
\end{enumerate}

\label{claim:tree}
    
\end{claim}

\begin{proof}
If $|S_1| \leq |S_2|$, by (P2), the agents in $S_2$ that swapped positions now have children of a different color than themselves, and still have utility $1$. The remaining agents in $S_2$ have no neighbors other than $\calA$ and retain utility $1$. Also, all the red children of $\calA$, have no neighbors other than $\calA$ and have utility $0$. See Figure~\ref{fig:swaps} (a). 

Otherwise, if $|S_1|>|S_2|$, all non-red children of $\calA$ have children of a different type than themselves and, by property (P1), have utility $1$. So, if $\calA$ moves from its position, these non-red agents would still have utility $1$. Some of the red children of $\calA $ may have red children, and others may have non-red children. See Figure~\ref{fig:swaps} (b).
\end{proof}

If there remains another mixed agent in the tree, then this agent $\calB$ is not a child of $\calA$, and we perform the same pre-processing operation for $\calB$'s children as well.

\begin{figure}[t]
\centering
\includegraphics[scale=0.38]{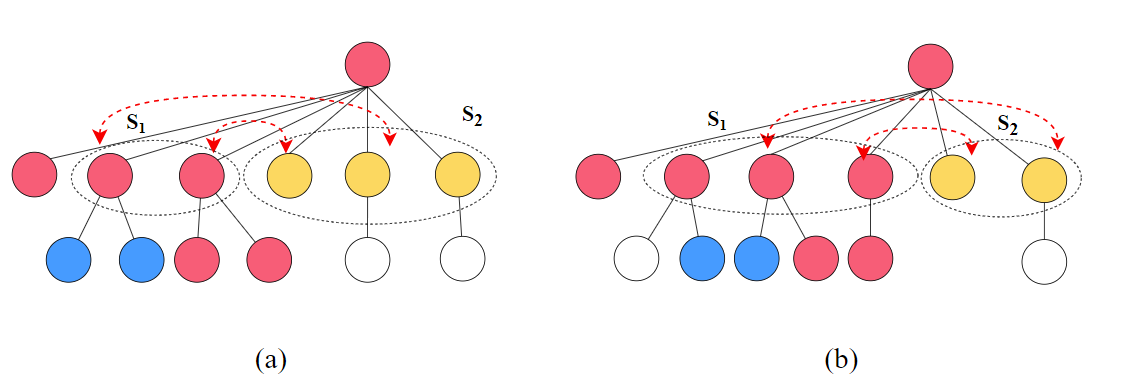}
\caption{ Swapping the red children adjacent to at least one more agent ($S_1$) with non-red children with no other agents as neighbors ($S_2$). (a): $|S_1| \leq |S_2|$. (b): $|S_1| > |S_2|$. Red dotted arrows show the swaps.}
\label{fig:swaps}
\end{figure}

We now create an ordered list of {\em candidate} red agents to move: starting with level $1$, moving down level by level, we put in red agents that have utility $0$ in this list. It remains to consider agents of utility strictly between $0$ and $1$: these are either agents that are mixed, or their children red agents who have non-red children. 

First consider agent $\calA$. If Condition (1) holds, then all of $\calA$'s red children are already in the list of candidates. If they all move, then $\calA$'s utility will change to $1$, so we do not put $\calA$ into the candidate list. Otherwise, if Condition (2) of Claim~\ref{claim:tree} holds, we put $\calA$ in the list of candidates. 
Next, suppose there was a second mixed agent $\calB$. As observed earlier, $\calB$ cannot be a child of $\calA$. 
Because our placement algorithm for the tree $G'$ placed agents from left to right, it can be verified that all of $\calA$'s red children can only have red children (see Figure~\ref{fig:two-mixed-agents}). Therefore, condition (1) of Claim~\ref{claim:tree} is satisfied, and $\calA$ was not placed in the candidate list. We follow the same procedure to decide if $\calB$ should be placed in the candidate list. We see that there is at most one mixed agent in the candidate list, and it is the last agent in the list. 

Next we describe {\em where} to move the candidate agents. Before we move any agent, note that all empty nodes have the property that they are adjacent to at most a single node containing an agent. This follows from the manner in which $G'$ was extracted from $G$ by repeatedly removing leaves. Therefore, all empty nodes either offer utility $0$ or $1$ to a red agent. When moving agents, we will maintain this property of the set of empty nodes. 

Call a node {\em available} if it is empty and adjacent only to a non-red agent, and perhaps other empty nodes. Observe that if a candidate red agent moves to an available node, we have the following properties:
(a) the utility of the candidate increases to $1$; 
(b) the utility of all other agents stays the same or increases; 
(c) the number of available nodes decreases by $1$;
(d) the newly vacated node is adjacent only to red agents and possibly empty nodes\footnote{Property (d) is not true if a mixed agent vacates its spot, but this happens only in the last step, when there are no more candidates.} and therefore offers utility $0$ to any red agent;
(e) any empty neighbor of the newly occupied node offers utility $0$ to any red agent. 
Thus, we maintain the desired property of empty nodes. 

We are now ready to describe our procedure to convert $C_G$ into an equilibrium assignment. We repeatedly take the first candidate from the ordered list of candidates and move it to an available node. This process ends when there are no more available nodes or there are no more candidates. If there are no more available nodes, we have an equilibrium assignment, as all other empty nodes offer utility $0$ to a red agent. If there are no more candidates, then all red agents have utility $1$, and we have an equilibrium assignment. 
\end{proof}

\begin{figure}[t]
\centering
\includegraphics[scale=0.45]{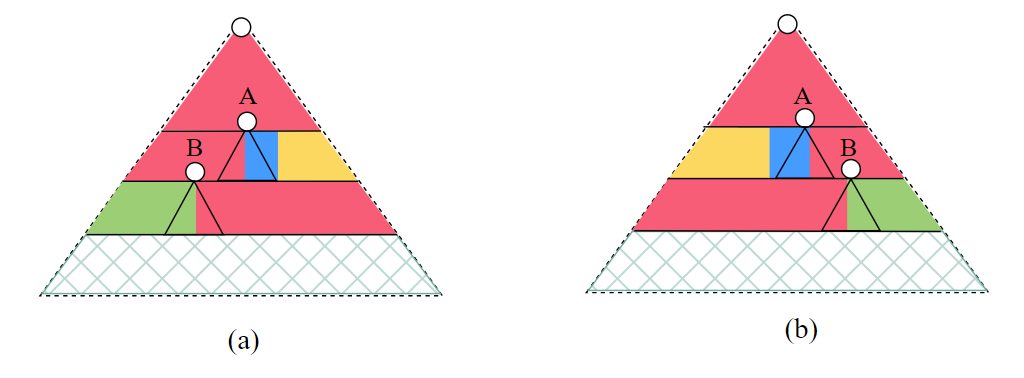}
\caption{All the red children of $\calA$ can only have red children (if any) : (a) $\calB$ is on the even level at which we finish placing the red agents, then all the red agents that are to the right of $\calB$ on the same level (thus children of $\calA$) can only have red children (if any) (b) $\calB$ is on the odd level at which we start placing the red agents, then all the red agents that are to the left of $\calB$ on the same level (thus children of $\calA$) can only have red children (if any). The cross hashed levels may or may not exist, but if they do, they are  occupied by non-red agents.}
\label{fig:two-mixed-agents}
\end{figure}

\subsection{Spider Topologies}
\label{subsec:Spider-graphs}
A {\em spider graph}, also called a starlike graph, is a tree in which all nodes besides one, called the {\em center node} $c$, have degree of at most 2 (see Figure~\ref{fig:cycle-spider-similarity-seeking} for an example); the degree of $c$ is $\delta(c) \geq 3$.
We will show that any \game in which the topology is a spider graph with a single empty node is a potential game. 
Unfortunately the function $\Phi$ from Definition~\ref{def:potential-function} is not a generalized ordinal potential function for any value of $m \in (0,1)$, as there are some improving moves for agents that increases the value of the function. However, we will show that any {\em long enough sequence of improving moves} is guaranteed to result in a lower value of the function. This will suffice to show convergence of IRD. 

We start with a lemma providing upper bounds on the change in potential caused by different types of moves.

\begin{lemma}\label{lem:spider-transitions}
For $k \geq 2$, in every \game $(R, S, T, G, \lambda)$~where $G$ is a spider graph, for an improving move of an agent $\calA$ of type $T$ in assignment $C_0$ that results in assignment $C_1$, we have:
 \begin{enumerate}
     \item $\Phi(C_1)-\Phi(C_0) \leq -m$ if the improving move does not involve the center node;
     \item $\Phi(C_1)-\Phi(C_0) \leq 2m-1 + n_T(c, C_1) -m\delta(c) $ if the improving move is from a degree-2 node to the center node;
     \item $\Phi(C_1)-\Phi(C_0) \leq m - m\delta(c) + n_T(c, C_0)$ if the improving move is from a leaf to the center node;
     \item $\Phi(C_1)-\Phi(C_0) \leq 1-2m + m\delta(c) - n_T(c, C_0)$ if the improving move is from the center node to any other node.
 \end{enumerate}
\end{lemma}

\begin{proof}
We consider the following exhaustive cases:
\begin{enumerate}
\item The improving move does not involve the center node $c$. 
\begin{enumerate}
    \item Agent $\calA$ jumps from a degree-2 node to another degree-2 node:
    \begin{enumerate}
    \item $u_{\calA}(C_0) = 0 $ and $u_{\calA}(C_1) = \frac{1}{2}$. Then $\Phi(C_1) - \Phi(C_0) = (1+2m) - (2+2m) = -1$.
    \item $u_{\calA}(C_0) = \frac{1}{2} $ and $u_{\calA}(C_1) = 1$. Then $\Phi(C_1) - \Phi(C_0) = (2m) - (1+2m) = -1$.
    \item $u_{\calA}(C_0) = 0 $ and $u_{\calA}(C_1) = 1$ and source and target nodes are not adjacent. Then $\Phi(C_1) - \Phi(C_0) = (2m) - (2+2m) = -2$.
    \item $u_{\calA}(C_0) = 0 $ and $u_{\calA}(C_1) = 1$ and source and target nodes are adjacent. Then $\Phi(C_1) - \Phi(C_0) = (2m) - (2m+1) = -1$.
    \end{enumerate}
    \item Agent $\calA$ jumps from degree-2 node to leaf node:
    \begin{enumerate}
    \item $u_{\calA}(C_0) = 0 $ and $u_{\calA}(C_1) = 1$. Then $\Phi(C_1) - \Phi(C_0) = (2m) - (2+m) = m-2$.
    \item $u_{\calA}(C_0) = \frac{1}{2} $ and $u_{\calA}(C_1) = 1$. Then $\Phi(C_1) - \Phi(C_0) = (2m) - (1+m) = m-1$.
    \end{enumerate}
    \item Agent $\calA$ jumps from leaf node to degree-2 node:
    \begin{enumerate}
    \item $u_{\calA}(C_0) = 0 $ and $u_{\calA}(C_1) = \frac{1}{2}$. Then $\Phi(C_1) - \Phi(C_0) = (m+1) - (1+2m) = -m$.
     \item $u_{\calA}(C_0) = 0 $ and $u_{\calA}(C_1) = 1$ and source and target nodes are not adjacent. Then $\Phi(C_1) - \Phi(C_0) = (m) - (1+2m) = -1 -m $.
     \item $u_{\calA}(C_0) = 0 $ and $u_{\calA}(C_1) = 1$ and source and target nodes are adjacent. Then $\Phi(C_1) - \Phi(C_0) = (m) - (2m) = -m$.
     \end{enumerate}
     \item Agent $\calA$ jumps from leaf node to leaf node: $\Phi(C_1) - \Phi(C_0) = (m) - (1+m) = -1$.
     
\end{enumerate}
\item The improving move is from a degree-2 node to center node:
\begin{enumerate}
    \item $u_{\calA}(C_0) = \frac{1}{2} $ and $u_{\calA}(C_1) >\frac{1}{2}$: $\Phi(C_1) - \Phi(C_0) = (2m + n_{T}(c, C_1)) - (1+m\delta(c))$.
     \item $u_{\calA}(C_0) = 0 $ and $u_{\calA}(C_1) > 0$ and source and target nodes are not adjacent: $\Phi(C_1) - \Phi(C_0) = (2m + n_{T}(c, C_1)) - (2+m\delta(c)) $.
    \item $u_{\calA}(C_0) = 0 $ and $u_{\calA}(C_1) > 0$ and source and target nodes are adjacent: $\Phi(C_1) - \Phi(C_0) = (2m + n_{T}(c, C_1)) - (1+m\delta(c))$.
\end{enumerate}
  \item The improving move is from a leaf to center node:
    \begin{enumerate}
    \item Source and target nodes are not adjacent: $\Phi(C_1) - \Phi(C_0) = (m + n_{T}(c, C_1)) - (1+m\delta(c)) $.
    \item Source and target nodes are adjacent: $\Phi(C_1) - \Phi(C_0) = (m + n_{T}(c, C_1)) - (m\delta(c))$.
    \end{enumerate}
\item  The improving move is from center node to any other node:
    \begin{enumerate}
    \item Agent $\calA$ jumps from center node to degree-2 node:
    \begin{enumerate}
    \item $u_{\calA}(C_0) <\frac{1}{2} $ and $u_{\calA}(C_1) = \frac{1}{2}$. Then $\Phi(C_1) - \Phi(C_0) = (1+ m\delta(c)) - (n_T(c, C_0) + 2m) $.
     \item $u_{\calA}(C_0) <1 $ and $u_{\calA}(C_1) = 1$ and source and target nodes are not adjacent. Then $\Phi(C_1) - \Phi(C_0) = (m\delta(c)) - (n_{T}(c, C_0) + 2m) $.
    \item $u_{\calA}(C_0) <1 $ and $u_{\calA}(C_1) = 1$ and source and target nodes are adjacent. Then $\Phi(C_1) - \Phi(C_0) =  (m\delta(c)) - (n_{T}(c, C_0) + 2m) $.
    \end{enumerate}
    \item Agent $\calA$ jumps from center node to leaf node: $\Phi(C_1) - \Phi(C_0) = (m\delta(c)) - (n_T(c, C_0) + m)$.
    \end{enumerate}
\end{enumerate}
Tables~\ref{tab:spider-transitions} and \ref{tab:spider-transitions-center} provide a more detailed counting of the change in potential based on the source and target node of a move.
\end{proof}

\begin{table}[h]
\small
\centering
\begin{tabular}{|c|c|c|c|}
\hline
Type of move                  & Source utility & Target utility & $\Delta \Phi \leq$ \\ \hline
\multirow{3}{*}{degree 2 to degree 2} & $0$              & $1/2$           & $-1$                   \\ \cline{2-4} 
                              & $0$              & $1$              & $-1$                   \\ \cline{2-4} 
                              & $1/2 $           & $1$              & $-1$                   \\ \hline
\multirow{2}{*}{degree 2 to leaf} & $0$              & $1$              & $m-2$                  \\ \cline{2-4} 
                              & $1/2$            & $1$              & $m-1 $                 \\ \hline
\multirow{2}{*}{Leaf to degree 2} & $0$              & $1/2$            & $-m$                   \\ \cline{2-4} 
                              & $0$              & $1$              & $-m$                   \\ \hline
Leaf to leaf                  & $0$              & $1$              & $-1$                   \\ \hline
\end{tabular}
\caption{Summary of the change in the potential change when the topology is a spider graph, resulting from improving moves that do not involve the center node. Note that $m < 1/2$.}
\label{tab:spider-transitions}
\end{table}

\begin{table*}[h]
\small
\centering
\begin{tabular}{|c|c|c|c|}
\hline
Type of move                    & Source utility & Target utility    & $\Delta \Phi \leq$ \\ \hline
degree 2 to center (adjacent) & $0$              & $>0$   &    $2m-1 + n_T(c, C_1) - m\delta(c) $              \\ \hline
\multirow{2}{*}{degree 2 to center (non-adjacent)} & $0$              & $>0$   &    $2m-2 + n_T(c, C_1) - m\delta(c) $            \\ \cline{2-4} 
                                & $1/2$            & $>1/2$ &
                                $2m - 1 + n_T(c, C_1) -m\delta(c)$             \\ \hline
\multirow{2}{*}{Center to degree 2} & $<1/2$ & $1/2$              &   $1 - 2m + m\delta(c) - n_T(c, C_0)$               \\ \cline{2-4} 
                                & $<1$   & $1$                 &        $-2m + m\delta(c) - n_T(c, C_0)$          \\ \hline
Leaf to center (adjacent) & $0$              & $>0$   &    $m + n_T(c, C_1) - m\delta(c)$                    \\ \hline
Leaf to center (non-adjacent)   & $0$              & $>0$   &    $m  -1 + n_T(c, C_1) - m\delta(c)$                    \\ \hline
Center to leaf             & $<1$   & $1$                 &      $- m + m\delta(c) - n_T(c, C_0)$            \\ \hline
\end{tabular}
\caption{ Summary of change in when the topology is a spider graph, resulting from improving moves that involve the center node. Note that  $m < 1/2$.}
\label{tab:spider-transitions-center}
\end{table*}

For the remainder of the section we will assume that $m < \frac{1}{2}$. Lemma~\ref{lem:spider-transitions} then implies that any improving move of an agent that does not involve the center node always {\em decreases} the value of the potential function $\Phi$. However, an improving move of an agent that moves to or out of the center can increase the potential by a non-constant amount.  The change in potential is related to the degree of the center node and the number of neighbors of the same type as the agent that moved to or from the center node. We claim, however, that such a move that increases the potential must be followed by moves that collectively decrease the potential {\em below} its value before the move in question, and thus that the game is a potential game. 

Suppose towards a contradiction that there is an IRC. Consider the smallest IRC of length $p$, and denote the assignments that it involves by  $C_0, C_1, \ldots C_{p-1}$.  Since any improving move of an agent that does not involve the center node always  decreases the potential,  our IRC must involve moves to or out of the center node. Since every move of an agent $\calA$ to the center node must be followed by a move of agent $\calA$ out of the center node, we calculate the change in potential resulting from the set of moves that includes the move to the center, move out of the center, and any moves in between. Let $C_{j_{0}}, C_{j_{1}}, \ldots C_{j_{w}}$, be the assignments in the cycle in which the center node $c$ is unoccupied, listed in increasing order. Furthermore, define  $C^{+}_{j_{i}}, C^{-}_{j_{i}},$ where $0 \leq i \leq w $ to be the immediate successor and predecessor of $C_{j_{i}}$. We will show that for every $i$, the difference in potential between $C_{j_i}$ and $C_{j_{i+1}}$ is always at most $0$; For readability, we use $i+1$ to mean $(i+1) \bmod (w+1)$ when referring to the indices of the assignments. Furthermore, we will show that it impossible for  the difference to be $0$ for {\em every such pair}, which yields a contradiction to the existence of an IRC. 

The following  lemma gives an upper bound on the change in potential caused by moves {\em inbetween} a move of an agent $\calA$ to the center node and the subsequent move out of the center node, that is between the assignments $C^+_{j_{i}}$ and $C^-_{j_{i+1}}$, for any $0 \leq i\leq w$. Let the center node $c$ be occupied by an agent $\calA$ of type $T$ in both these assignments and all assignments in between. Assume that $\calA$ moved from a source node $s$ in $C_{j_i}$ to the center node and subsequently moved to a target node $t$ in $C_{j_{i+1}}$. Let $\Delta_T = n_T(c,C^-_{j_{i+1}})- n_T(c, C^+_{j_{i}})$. 

\begin{lemma}\label{lem:change-in-neighborhood}
We have
\begin{numcases}{ \Phi(C^-_{j_{i+1}}) - \Phi(C^+_{j_{i}}) \leq}
 - |\Delta_{T}| & \text{if $s$, $t$ are not adjacent to $c$} \label{case-st}\\
 -|\Delta_{T}| + m &\text{ if only $t$ is adjacent to $c$} \label{case-t}\\
 -|\Delta_{T}| + m -1&\text{if only $s$ is adjacent to $c$}. \label{case-s}
\\
-|\Delta_{T}| + 2m -1 &\text{if $s$, $t$ are adjacent to $c$}. \label{case-no-st}
\end{numcases}
\end{lemma}

\begin{proof}
We will start with the easiest case when neither $s$ nor $t$ is adjacent to the center. It is clear that  every change in the number of neighbors of the center that are occupied by agents of type $T$ requires a pair of improving moves, an agent of type $T$ to move out of the neighborhood and an agent of some other type to move to its place (thereby decreasing the number of neighbors of type $T$), or an agent of a different type to move out of the neighborhood and an agent of type $T$ to move to its place (increasing the number of neighbors of type $T$). 
From Table ~\ref{tab:spider-transitions}, we get an upper bound on the change in potential of each such pair of moves. First we note that it is impossible for both such moves to be (leaf to degree $2$), or both such moves to be (degree $2$ to leaf). It is straightforward to verify that all other pairs of moves have a net change in potential of at least $-1$. Inequality \eqref{case-st} follows. 

Now suppose that $s$ is not adjacent to the center, but $t$ is. 
Then, the agent $\calB$ at node $t$ in the assignment before $C^-_{j_{i+1}}$ moved out to leave $t$ vacant in $C^-_{j_{i+1}}$. 
If $\calB$ is not of type $T$, it did not contribute to a change in the number of neighbors of type $T$ of the center, and the same reasoning as in the previous case applies to give us Inequality~\eqref{case-st}, which implies Inequality~\eqref{case-t}. 
If $\calB$ instead is of type $T$, then its move out of $t$ does change the number of neighbors of type $T$, but does not have a matching move back to the center. Noting that $t$ cannot be a leaf since $\calA$ would not jump to it in the next step, the move of $\calB$ out of $t$ causes a potential change of at most $m-1$ (using Table~\ref{tab:spider-transitions}). Together with the potential change due to the remaining $\Delta_T-1$ pairs of improving moves, we get Inequality~\eqref{case-t}. 

If $s$ is adjacent to the center and $t$ is not, observe that the agent that moves to $s$ immediately after $\calA$ moves to the center is not of type $T$ and does not contribute to the change in neighbors of type $T$. The change in the potential resulting from this improving move is at most $m-1$ as shown in Table~\ref{tab:spider-transitions}, and together with the change in potential caused by $\Delta_T$ pairs of improving moves as in the first case, we get Inequality~\eqref{case-s}.
        
Last we consider the case when both $s$ and $t$ are adjacent to the center. The first move does not contribute to the change in neighbors of type $T$ and changes the potential by at most $m-1$, while the last move may contribute to this change, and in that case, would not have a matching move of an agent back to a neighbor of the center. The total change in potential is at most $m-1$ for the first move, at most $m-1$ for the last move, and $|\Delta_T|-1$ for pairs of improving moves contributing to change in neighbors of type $T$. Thus, we get Inequality~\eqref{case-no-st}.
\end{proof}

Now we consider the difference in potential between $C_{j_{i+1}}$ and $C_{j_i}$.
We analyze separately the cases when $s$ is a degree-2 node or a leaf node.

\begin{lemma}
\label{lem:jump-from-degree-2}
If $s$ is a degree-2 node, then $\Phi(C_{j_{i+1}}) - \Phi(C_{j_i}) < 0$
\end{lemma}

\begin{proof} 
The source utility of $\calA$ at node $s$ could be either $0$ or $1/2$ and the target utility at node $t$ could be $1/2$ or $1$. 
We consider all possible combinations below:

\begin{enumerate}

    \item 
          $u_{\calA}(C_{j_i}) = 0$, $u_{\calA}(C_{j_{i+1}}) = 1/2$ and the source node $s$ {\em is not} adjacent to the center. 
          In assignment $C_{j_i}$, agent $\calA$ is located at $s$ and has two neighbors of the same type. 
          Since $u_{\calA}(C_{j_i}) = 0$, $u_{\calA}(C^+_{j_i}) > 0$. 
          Eventually agent $\calA$ moves out of the center node to the target node $t$ to reach assignment $C_{j_{i+1}}$ where it has a neighbor of the same type and a neighbor of different type. 
          Note that $t$ must be a degree-$2$ node and it cannot be adjacent to the center node. 
          Since $u_{\calA}(C_{j_{i+1}}) = 1/2$, $u_{\calA}(C^-_{j_{i+1}}) < 1/2$. Therefore, in assignment $C^-_{j_{i+1}}$, more than half of the neighbors of agent $\calA$ are of type $T$. From Table ~\ref{tab:spider-transitions-center}, we obtain:
          \begin{align}
          \Phi(C^+_{j_{i}}) - \Phi(C_{j_{i}}) &\leq -m\delta(c) + n_T(c, C^+_{j_{i}}) +2m-2  \label{eq:to-center-case1}
          \end{align}
          and 
          \begin{align}
            \Phi(C_{j_{i+1}}) - \Phi(C^-_{j_{i+1}}) &\leq m\delta(c) - n_T(c, C^-_{j_{i+1}}) + 1- 2m. \label{eq:from-center-case1}
          \end{align}
        Then, since any improving move between $C^+_{j_{i}}$ and $C^-_{j_{i+1}}$ decreases the potential as shown in Lemma~\ref{lem:spider-transitions}, we obtain an upper bound on the net potential change from $C_{j_i}$ to
        $C_{j_{i+1}}$ by adding Inequalities~\eqref{case-st},~\eqref{eq:to-center-case1}, and~\eqref{eq:from-center-case1}: 
         \begin{align*}
            \Phi(C_{j_{i+1}}) - \Phi(C_{j_{i}}) 
            &\leq n_T(c, C^+_{j_{i}}) - n_T(c, C^-_{j_{i+1}}) - 1 - |(n_T(c, C^+_{j_{i}}) - n_T(c, C^-_{j_{i+1}}))|  < 0.
        \end{align*}

        \item $u_{\calA}(C_{j_i}) = 0$, $u_{\calA}(C_{j_{i+1}}) = 1/2$, and the source node $s$ {\em is} adjacent to the center. 
        In assignment $C_{j_i}$, agent $\calA$ is located at $s$ and is adjacent to an agent of the same type and the empty center node. 
        Since $u_{\calA}(C_{j_i}) = 0$, $u_{\calA}(C^+_{j_i}) > 0$. Eventually agent $\calA$ moves out of the center node to node $t$ to reach assignment $C_{j_{i+1}}$ where it has a neighbor of the same type and a neighbor of different type. Note that $t$ must be a degree-2 node and it cannot be adjacent to the center node. Since $u_{\calA}(C_{j_{i+1}}) = 1/2$, $u_{\calA}(C^-_{j_{i+1}}) < 1/2$.
        Therefore, in assignment $C^-_{j_{i+1}}$, more than half of the neighbors of agent $\calA$ are of type $T$. 
        From Table ~\ref{tab:spider-transitions-center}, we obtain:
        \begin{align}
        \label{eq:to-center-case1b}
            \Phi(C^+_{j_{i}}) - \Phi(C_{j_{i}}) &\leq -m\delta(c) + n_{T_x}(c, C^+_{j_{i}})+2m-1 
        \end{align}
        and
        \begin{align}
        \label{eq:from-center-case1b}
            \Phi(C_{j_{i+1}}) - \Phi(C^-_{j_{i+1}}) &\leq m\delta(c) - n_{T_x}(c, C^-_{j_{i+1}})+ 1- 2m.
        \end{align}
       Then, since any improving move between $C^{+}_{j_{i}}$ and $C^-_{j_{i+1}}$ decreases the potential as shown in Lemma~\ref{lem:spider-transitions}, we obtain an upper bound on the net potential change from $C_{j_i}$ to
        $C_{j_{i+1}}$ by adding Inequalities~\eqref{case-s},~\eqref{eq:to-center-case1b}, and~\eqref{eq:from-center-case1b}: 
         \begin{align*}
            \Phi(C_{j_{i+1}}) - \Phi(C_{j_{i}}) &\leq n_{T}(c, C^+_{j_{i}}) - n_{T}(c, C^-_{j_{i+1}}) + m -1-|(n_{T}(c, C^+_{j_{i}}) - n_{T_x}(c, C^-_{j_{i+1}}))| < 0.
        \end{align*}
        Note that, if $n_{T}(c,C^+_{j_{i}}) > n_{T}(c, C^-_{j_{i+1}})$, then the number of agents of type $T$ that are adjacent to the center node should have decreased, i.e., replaced by agents of other types. Each replacement requires a pair of moves: a move out of the neighborhood by an agent of type $T$ and a move into the neighborhood by an agent of other type. As such, the agent that moves to $s$ after $\calA$ moves to the center is not part of the replacement pair and causes an additional decrease in the potential. Therefore, the net potential is negative.
        
       \item $u_{\calA}(C_{j_i}) = 0$ and  $u_{\calA}(C_{j_{i+1}}) = 1$.
       In assignment $C_{j_i}$, agent $\calA$ is located at $s$ and either both its neighbors are of same type, or it is adjacent to an agent of the same type and the empty center node. Since $u_{\calA}(C_{j_i}) = 0$, $u_{\calA}(C^+_{j_i}) > 0$. Eventually agent $\calA$ moves out of the center node to node $t$ to reach assignment $C_{j_{i+1}}$ where it has a neighbor of different type. Since $u_{\calA}(C_{j_{i+1}}) = 1$, $u_{\calA}(C^-_{j_{i+1}}) < 1$. From Table ~\ref{tab:spider-transitions-center}, we obtain:
        \begin{align}
        \label{eq:to-center-case1c}
            \Phi(C^+_{j_{i}}) - \Phi(C_{j_{i}}) &\leq -m\delta(c) + n_{T}(c, C^+_{j_{i}})+2m-1 
        \end{align}
        and 
        \begin{align}
        \label{eq:from-center-case1c}
            \Phi(C_{j_{i+1}}) - \Phi(C^-_{j_{i+1}}) &\leq m\delta(c) - n_{T}(c, C^-_{j_{i+1}})-m.
        \end{align}
        Since any improving move between $C^+_{j_{i}}$ and $C^-_{j_{i+1}}$ decreases the potential as shown in Lemma~\ref{lem:spider-transitions}, we obtain an upper bound on the net potential change from $C_{j_i}$ to
        $C_{j_{i+1}}$ by adding Inequalities~\eqref{case-t},~\eqref{eq:to-center-case1c}, and~\eqref{eq:from-center-case1c}: 
         \begin{align*}
            \Phi(C_{j_{i+1}}) - \Phi(C_{j_{i}}) &\leq n_{T}(c, C^+_{j_{i}}) - n_{T}(c, C^-_{j_{i+1}}) + 2m-1 
            - |(n_{T}(c, C^+_{j_{i}}) - n_{T}(c, C^-_{j_{i+1}}))|
            < 0.
        \end{align*}
        
         \item $u_{\calA}(C_{j_i}) = 1/2$ and $u_{\calA}(C_{j_{i+1}}) = 1/2$. 
         In assignment $C_{j_i}$, agent $\calA$ while located at $s$ has a neighbor of the same type and a neighbor of different type. Since $u_{\calA}(C_{j_i}) = 1/2$, $u_{\calA}(C^+_{j_i}) > 1/2$. Therefore, in assignment $C^+_{j_{i}}$, less than half of the neighbors of agent $\calA$ are of type $T$. Eventually agent $\calA$ moves out of the center node to the target node $t$ to reach assignment $C_{j_{i+1}}$ where it has a neighbor of the same type and a neighbor of different type. Note that $t$ must be a degree-2 node. Since $u_{\calA}(C_{j_{i+1}}) = 1/2$, $u_{\calA}(C^-_{j_{i+1}}) < \frac{1}{2}$. Hence, in assignment $C^-_{j_{i+1}}$, more than half of the neighbors of agent $\calA$ are of type $T$, and we obtain
        \begin{align}
        \label{eq:neighbor-case1a}
           n_{T}(c, C^+_{j_{i}}) < n_{T}(c, C^-_{j_{i+1}}).
        \end{align}
          We claim that the net potential change is negative for improving moves that leads from $C_{j_{i}}$ to $C^+_{j_{i}}$, from $C^-_{j_{i+1}}$ to $C_{j_{i+1}}$, and from $C^+_{j_{i}}$ to $C^-_{j_{i+1}}$. Since agent $\calA$ must improve its utility from $1/2$ in $C_{j_i}$ to strictly more than $1/2$ in $C^+_{j_{i}}$, from Table~\ref{tab:spider-transitions-center}, we obtain:
        \begin{align}
        \label{eq:to-center-case1a}
            \Phi(C^+_{j_{i}}) - \Phi(C_{j_{i}}) &\leq -m\delta(c) + n_{T}(c, C^+_{j_{i}})+2m-1.
        \end{align}
        Similarly, since agent $\calA$ must improve its utility from strictly less than $1/2$ in $C^-_{j_{i+1}}$ to $1/2$ in $C_{j_{i+1}}$, from Table~\ref{tab:spider-transitions-center}, we obtain:
        \begin{align}
        \label{eq:from-center-case1a}
            \Phi(C_{j_{i+1}}) - \Phi(C^-_{j_{i+1}}) &\leq m\delta(c) - n_{T}(c, C^-_{j_{i+1}})+ 1- 2m.
        \end{align}
        So, since all any improving move between $C^+_{j_{i}}$ and $C^-_{j_{i+1}}$, decreases the potential as shown in Lemma~\ref{lem:spider-transitions}, we obtain an upper bound on the net potential change from $C_{j_i}$ to
        $C_{j_{i+1}}$ by adding Inequalities~\eqref{case-st},~\eqref{eq:to-center-case1a}, and~\eqref{eq:from-center-case1a}: 
         \begin{align*}
            \Phi(C_{j_{i+1}}) - \Phi(C_{j_{i}}) &\leq n_{T}(c, C^+_{j_{i}}) - n_{T}(c, C^-_{j_{i+1}}) -|(n_{T}(c, C^+_{j_{i}}) - n_{T}(c, C^-_{j_{i+1}}))| < 0
        \end{align*}
        where the last inequality follows from Inequality ~\eqref{eq:neighbor-case1a}. 
        
        \item $u_{\calA}(C_{j_i}) = 1/2$ and $u_{\calA}(C_{j_{i+1}}) = 1$. In assignment $C_{j_i}$, agent $\calA$ while located at $s$ has a neighbor of the same type and a neighbor of different type. Since $u_{\calA}(C_{j_i}) = \frac{1}{2}$, $u_{\calA}(C^+_{j_i}) > \frac{1}{2}$. Eventually agent $\calA$ moves out of the center node to node $t$ to reach assignment $C_{j_{i+1}}$ where it has a neighbor of different type. Since $u_{\calA}(C_{j_{i+1}}) = 1$, $u_{\calA}(C^-_{j_{i+1}}) < 1$. From Table ~\ref{tab:spider-transitions-center}, we obtain:
        \begin{align}
        \label{eq:to-center-case1d}
            \Phi(C^+_{j_{i}}) - \Phi(C_{j_{i}}) &\leq -m\delta(c) + n_{T}(c, C^+_{j_{i}})+2m-1 
        \end{align}
        and 
        \begin{align}
        \label{eq:from-center-case1d}
            \Phi(C_{j_{i+1}}) - \Phi(C^-_{j_{i+1}}) &\leq m\delta(c) - n_{T}(c, C^-_{j_{i+1}}) -m.
        \end{align}
        Since any improving move between $C^+_{j_{i}}$ and $C^-_{j_{i+1}}$ decreases the potential as shown in Lemma~\ref{lem:spider-transitions}, we obtain an upper bound on the net potential change from $C_{j_i}$ to
        $C_{j_{i+1}}$ by adding Inequalities~\eqref{case-t},~\eqref{eq:to-center-case1d}, and~\eqref{eq:from-center-case1d}: 
         \begin{align*}
            \Phi(C_{j_{i+1}}) - \Phi(C_{j_{i}}) &\leq n_{T}(c, C^+_{j_{i}}) - n_{T}(c, C^-_{j_{i+1}})+2m-1
            - |(n_{T}(c, C^+_{j_{i}}) - n_{T}(c, C^-_{j_{i+1}}))|
            < 0.
        \end{align*}
    \end{enumerate}
This completes the proof for the case when there exists an assignment $C_{j_i}$ where $\calA$ occupies a degree-2 node.
\end{proof}

\begin{lemma}
\label{jump-from-leaf}
If $s$ is a leaf node, then $\Phi(C_{j_{i+1}}) - \Phi(C_{j_i}) \leq 0$. 
Furthermore, $\Phi(C_{j_{i+1}}) - \Phi(C_{j_i}) = 0$ if $s$ is adjacent to the center and $u_{\calA}(C_{j_{i+1}}) = 1/2$.
\end{lemma}
\begin{proof}
 The source utility of $\calA$ at node $s$ should be 0 and the target utility at node $t$ could be $1/2$ or 1. 
 We consider the four possible cases below:
    \begin{enumerate}
     \item $u_{\calA}(C_{j_{i+1}}) = 1$ and $s$ is not adjacent to the center.
     In this case, agent $\calA$ moves out of the center node to node $t$ to reach assignment $C_{j_{i+1}}$ where it has a neighbor of different type. Since $u_{\calA}(C_{j_{i+1}}) = 1$, $u_{\calA}(C^-_{j_{i+1}}) < 1$. 
     From Table ~\ref{tab:spider-transitions-center}, we obtain:
         \begin{align}
         \label{eq:to-center-target-utility-one-case1}
            \Phi(C^+_{j_{i}}) - \Phi(C_{j_{i}}) &\leq -m\delta(c) + n_{T}(c, C^+_{j_{i}}) +m - 1
        \end{align}
        and
        \begin{align}
        \label{eq:from-center-target-utility-one-case1}
          \Phi(C_{j_{i+1}}) - \Phi(C^-_{j_{i+1}}) &\leq m\delta(c) - n_{T}(c, C^-_{j_{i+1}}) - m.
        \end{align}
        Since any improving move between $C^+_{j_{i}}$ and $C^-_{j_{i+1}}$ decreases the potential as shown in Lemma~\ref{lem:spider-transitions}, we obtain an upper bound on the net potential change from $C_{j_i}$ to
        $C_{j_{i+1}}$ by adding Inequalities~\eqref{case-t},~\eqref{eq:to-center-target-utility-one-case1}, and~\eqref{eq:from-center-target-utility-one-case1}: 
         \begin{align*}
           \Phi(C_{j_{i+1}}) - \Phi(C_{j_{i}}) &\leq n_{T}(c, C^+_{j_{i}}) - n_{T}(c, C^-_{j_{i+1}}) + m- 1
            - |(n_{T}(c, C^+_{j_{i}}) - n_{T}(c, C^-_{j_{i+1}}))|
            < 0.
        \end{align*} 
        
        \item $u_{\calA}(C_{j_{i+1}}) = 1$ and $s$ is adjacent to center. In this case, agent $\calA$ moves out of the center node to node $t$ to reach assignment $C_{j_{i+1}}$ where it has a neighbor of different type. Since $u_{\calA}(C_{j_{i+1}}) = 1$, $u_{\calA}(C^-_{j_{i+1}}) < 1$. From Table ~\ref{tab:spider-transitions-center}, we obtain:
         \begin{align}
         \label{eq:to-center-target-utility-one}
            \Phi(C^+_{j_{i}}) - \Phi(C_{j_{i}}) &\leq -m\delta(c) + n_{T}(c, C^+_{j_{i}}) +m
        \end{align}
        and 
        \begin{align}
        \label{eq:from-center-target-utility-one}
          \Phi(C_{j_{i+1}}) - \Phi(C^-_{j_{i+1}}) &\leq m\delta(c) - n_{T}(c, C^-_{j_{i+1}}) - m.
        \end{align}
        Since any improving move between $C^+_{j_{i}}$ and $C^-_{j_{i+1}}$ decreases the potential as shown in Lemma~\ref{lem:spider-transitions}, we obtain an upper bound on the net potential change from $C_{j_i}$ to
        $C_{j_{i+1}}$ by adding Inequalities~\eqref{case-no-st},~\eqref{eq:to-center-target-utility-one}, and~\eqref{eq:from-center-target-utility-one}: 
         \begin{align*}
           \Phi(C_{j_{i+1}}) - \Phi(C_{j_{i}}) &\leq n_{T}(c, C^+_{j_{i}}) - n_{T}(c, C^-_{j_{i+1}}) + 2m- 1
            - |(n_{T}(c, C^+_{j_{i}}) - n_{T}(c, C^-_{j_{i+1}}))|
            < 0.
        \end{align*} 
        
        \item $u_{\calA}(C_{j_{i+1}}) = 1/2$ and $s$ is not adjacent to center. In this case, agent $\calA$ moves out of the center node to node $t$ to reach assignment $C_{j_{i+1}}$ where it has one neighbor of the same type and one neighbor of different type. Since $u_{\calA}(C_{j_{i+1}}) =1/2$, $u_{\calA}(C^-_{j_{i+1}}) < \frac{1}{2}$. From Table ~\ref{tab:spider-transitions-center}, we obtain:
         \begin{align}
         \label{eq:to-center-nonadjacent-target-utility-half}
            \Phi(C^+_{j_{i}}) - \Phi(C_{j_{i}}) &\leq -m\delta(c) + n_{T}(c, C^+_{j_{i}}) +m -1 
        \end{align}
        and
        \begin{align}
        \label{eq:from-center-nonadjacent-target-utility-half}
          \Phi(C_{j_{i+1}}) - \Phi(C^-_{j_{i+1}}) &\leq m\delta(c) - n_{T}(c, C^-_{j_{i+1}})+1 - 2m.
        \end{align}
        Since any improving move between $C^+_{j_{i}}$ and $C^-_{j_{i+1}}$ decreases the potential as shown in Lemma~\ref{lem:spider-transitions}, we obtain an upper bound on the net potential change from $C_{j_i}$ to
        $C_{j_{i+1}}$ by adding Inequalities~\eqref{case-st},~\eqref{eq:to-center-nonadjacent-target-utility-half}, and~\eqref{eq:from-center-nonadjacent-target-utility-half}:  
         \begin{align}
         \label{eq:zero-net-nonadjacent}
          \Phi(C_{j_{i+1}}) - \Phi(C_{j_{i}}) &\leq n_{T}(c, C^+_{j_{i}}) - n_{T}(c, C^-_{j_{i+1}}) - m - |(n_{T}(c, C^+_{j_{i}}) - n_{T}(c, C^-_{j_{i+1}}))| < 0.
        \end{align}
        
         \item $u_{\calA}(C_{j_{i+1}}) = 1/2$ and $s$ is adjacent to center. In this case, agent $\calA$ moves out of the center node to node $t$ to reach assignment $C_{j_{i+1}}$ where it has one neighbor of the same type and one neighbor of different type. Since $u_{\calA}(C_{j_{i+1}}) = 1/2$, $u_{\calA}(C^-_{j_{i+1}}) < 1/2$. From Table ~\ref{tab:spider-transitions-center}, we obtain:
         \begin{align}
         \label{eq:to-center-adjacent-target-utility-half}
            \Phi(C^+_{j_{i}}) - \Phi(C_{j_{i}}) &\leq -m\delta(c) + n_{T}(c, C^+_{j_{i}}) +m
        \end{align}
        and
        \begin{align}
        \label{eq:from-center-adjacent-target-utility-half}
          \Phi(C_{j_{i+1}}) - \Phi(C^-_{j_{i+1}}) &\leq m\delta(c) - n_{T}(c, C^-_{j_{i+1}})+1 - 2m.
        \end{align}
        Since any improving move between $C^+_{j_{i}}$ and $C^-_{j_{i+1}}$ decreases the potential as shown in Lemma~\ref{lem:spider-transitions}, we obtain an upper bound on the net potential change from $C_{j_i}$ to
        $C_{j_{i+1}}$ by adding Inequalities~\eqref{case-s},~\eqref{eq:to-center-adjacent-target-utility-half}, and~\eqref{eq:from-center-adjacent-target-utility-half}:  
         \begin{align*}
         \label{eq:zero-net-adjacent}
          \Phi(C_{j_{i+1}}) - \Phi(C_{j_{i}}) &\leq n_{T}(c, C^+_{j_{i}}) - n_{T}(c, C^-_{j_{i+1}}) - |(n_{T}(c, C^+_{j_{i}}) - n_{T}(c, C^-_{j_{i+1}}))|\leq 0.
        \end{align*}
    \end{enumerate}
The proof is complete.
\end{proof}

We are now ready to prove the main result of this section.

\begin{theorem}\label{thm:spider}
For $k\geq2$, every \game $(R, S, T, G, \lambda)$ where $G$ is a spider graph and $|V| = |R \cup S| + 1$ is a potential game.
\end{theorem}

\begin{proof}
Suppose instead that there is a game that is not a potential game. Then, there exists an IRC, and the net potential change of the moves in this cycle must be $0$. By Lemma~\ref{lem:spider-transitions}, every improving move that does not involve the center node decreases the potential. It follows that the cycle must involve moves to and out of the center node. Let $C_{j_{0}}, C_{j_{1}}, \ldots C_{j_{w}}$, be the assignments in the cycle in which the center node $c$ is empty, listed in increasing order. Lemmas~\ref{lem:jump-from-degree-2} and ~\ref{jump-from-leaf} show that for every $i$, the difference in potential between $C_{j_i}$ and $C_{j_{i+1}}$ is at most $0$. Since the potential never increases, and the net potential change must be $0$, it must in fact be that, {\em for every} $i$, $\Phi(C_{j_{i+1}}) - \Phi(C_{j_i}) = 0$. It follows from Lemma~\ref{jump-from-leaf} that, for every $i$, an agent moves to the center node from a leaf node adjacent to the center, and subsequently moves out to a degree-$2$ node to get a utility $1/2$. 

Assume without loss of generality that an agent $\calA$ of  type $T$ makes a move from such a leaf node $s$ neighbor of the center to occupy the center in $C_{j_0}^+$. Before $\calA$ moves out of the center node, clearly no agent of type $T$ would move to $s$, as it would get a utility $0$. 
It follows that, if {\em every} agent that move to the center in the cycle is of the same type $T$, then the number of neighbors of the center of type $T$ must decrease during the moves comprising the cycle, which is a contradiction.

Therefore, it must be the case that, for some $j_i$, an agent $\calA$ of type $T$ occupies the center node in $C^+_{j_{i}}$, and an agent $\calB$ of type $T'\neq T$ occupies the center in $C^+_{j_{i+1}}$. By Lemma~\ref{jump-from-leaf}, we have $u_{\calA}(C_{j_{i+1}}) = 1/2$, therefore $u_{\calA}(C^-_{j_{i+1}}) < 1/2$. Since $\calB$ is of a different type to $\calA$, it must be that $u_{\calB}(C^+_{j_{i+1}}) > 1/2$. 

Since the change in potential between $C_{j_{i+1}}$ and $C_{j_{i+2}}$ is also $0$, by Lemma~\ref{jump-from-leaf}, it must be that $\calB$ moves to a node $t$ to get a utility $1/2$, and thus $u_{\calB}(C^-_{j_{i+2}}) < 1/2$. This means that the  number of agents of type $T'$ that occupy the nodes adjacent to the center must have increased, that is, $n_{T'}(c,C^+_{j_{i+1}}) - n_{T'}(c, C^-_{j_{i+2}})<0$. Now, as in the proof of case 3 in Lemma~\ref{jump-from-leaf}, we have:
\begin{align*}
    \Phi(C_{j_{i+2}}) - \Phi(C_{j_{i+1}}) &\leq n_{T'}(c, C^+_{j_{i+1}}) - n_{T'}(c, C^-_{j_{i+2}}) - |(n_{T'}(c, C^+_{j_{i+1}}) - n_{T'}(c, C^-_{j_{i+2}}))|< 0,
\end{align*} 
which is a contradiction. As such, there can be no cycle in the assignment graph, and we conclude that the game must be potential.
\end{proof}

We remark that our proof assumes that there is a single empty node. It is unclear what happens when there are more empty nodes in the graph. In contrast to the above result, the similarity-seeking jump games studied in \cite{ijcai2019p38} are {\em not} potential; see the example given in Figure~\ref{fig:cycle-spider-similarity-seeking}.

\begin{figure}[!htb]
    \centering
    \includegraphics[scale=0.45]{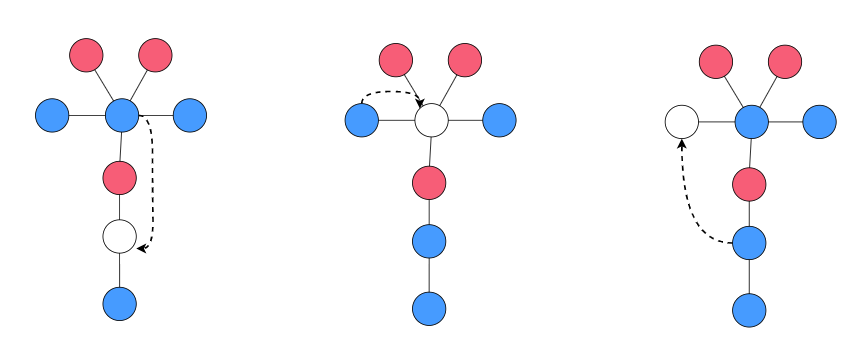}
    \caption{An IRC in a spider graph with similarity-seeking strategic agents.}
    \label{fig:cycle-spider-similarity-seeking}
\end{figure}

\section{Networks of Bounded Degree} \label{sec:bounded}
In this section, we consider \games on networks of fixed degree. We first show that \games are potential games whenever the topology is a regular graph and there is there is a single empty node. 

\begin{theorem}\label{thm:k-regular-graph}
For any $k, \delta \geq 2 $, every \game $(R, S, T, G, \lambda)$ where $G$ is a $\delta$-regular graph and $|V| = |R \cup S| + 1$ is a potential game.
\end{theorem}

\begin{proof}
Assume that there are no stubborn agents. We show that the function $\Phi$ from Definition ~\ref{def:potential-function} is a generalized ordinal potential function. Consider an improving move of an agent $\calA$ in assignment $C_0$ and call the resulting assignment $C_1$. If $\calA$ does not move to a neighboring node, then we have  $u_{\calA}(C_0) = x_0/\delta $ and $u_{\calA}(C_1) = x_1/\delta$, where $u_{\calA}(C_0) < u_{\calA}(C_1)$ and therefore $x_0 < x_1$. Hence,  
\begin{align*}
  \Phi(C_1) - \Phi(C_0) & = m\delta + (\delta- x_1) - (\delta - x_0) -m\delta =  x_0 - x_1 < 0.
  \end{align*}   
It is easy to verify that the difference in potential is the same if $\calA$ moves to a neighboring empty node. 
This shows that $\Phi$ is a generalized ordinal potential function (for any value of $m \in (0,1)$) which completes the proof for the case where there are no stubborn agents. By Lemma~\ref{lem:stubborn}, the game is also a potential game when $S \neq \varnothing$.
\end{proof}

Theorem~\ref{thm:k-regular-graph} does not generally extend to the case where there are more empty nodes in the topology. We show next that when there are $4$ or more empty nodes, there can be IRCs in the dynamics. 

\begin{theorem}\label{thm:k-regular-graph-multi-empty}
For every $\delta \geq 4$ and $k \geq 2$, there exists a \game $I=(R, S, T, G, \lambda)$ where $G$ is a $\delta$-regular graph and $|V| = |R \cup S| + \delta$ such that $I$ is not a potential game.
\end{theorem}

\begin{proof}
 For $\delta = 4$, we show an IRC in a regular graph of degree $4$ as depicted in Figure~\ref{fig:cycleRegular}. It is easy to verify that the moves shown are improving ones. Indeed, in the first move, the utility of the red agent improves from $0$ to $1/2$. In the second move, the utility of the blue agent improves from $3/4$ to $1$. In the third move, the utility of the red agent improves from $2/3$ to $3/4$. In the next three moves, the red and blue agents interchange roles, with the same changes in utilities. 
 
 We now give an inductive argument for the case of $\delta > 4$. Suppose inductively that there is an IRC in game $I$ on a $\delta$-regular graph $G$ for some $\delta \geq 4$. We now construct a $(\delta + 1)$-regular graph $H$ by creating a duplicate of graph $G$ called $G'$ and connecting every node in $G'$ to its counterpart in $G$. Consider now the $(\delta+1)$-regular graph $G \cup G'$, and the following placement of agents. We keep the same assignment of agents in $G$. The counterpart in $G'$ of the node containing the moving red agent has no agents assigned to it, while all other nodes in $G'$ are assigned yellow agents. It can be verified that the same IRC exists in this graph for the same red agent and blue agent visiting the same nodes as in $G$. Indeed, in the first move, the utility of the red agent improves from 0 to $\frac{d-3}{d-2}$. In the second move, the utility of the blue agent improves from $\frac{d-1}{d}$ to $1$. In the third move, the utility of the red agent improves from $\frac{d-2}{d-1}$ to $\frac{d-1}{d}$. As before, in the next three moves, the red and blue agents interchange roles, with the same changes in utilities. This completes the inductive argument. 
\end{proof}

\begin{figure}[t]
    \centering
    \includegraphics[scale=0.5]{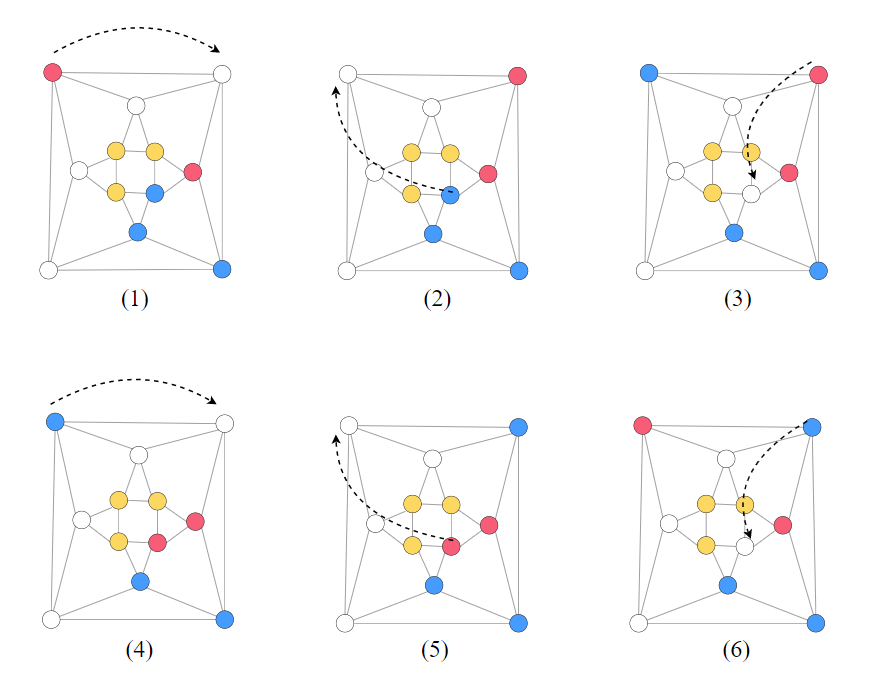}
    \caption{An IRC in a regular graph of degree 4 with 4 empty nodes.}
    \label{fig:cycleRegular}
\end{figure}
    
Note that there is a gap between the results of Theorems~\ref{thm:k-regular-graph} and \ref{thm:k-regular-graph-multi-empty}. The question of whether \games are potential games in regular graphs of degree 3, or regular graphs of degree 4 with 2 or 3 empty nodes remains open. Finally, we remark that an IRC on a regular graph with $k=2$ can be constructed from the tree shown in Figure~\ref{fig:cycleTree} but the number of empty nodes would be greater than $4$. 

Next, we consider topologies of maximum degree 2 (that may not be regular), and show that, regardless of the number of empty nodes, any \game is a potential game. 

\begin{theorem}\label{thm:degree-2}
For $k\geq2 $, every \game $(R, S, T, G, \lambda)$ where $G$ is a graph of maximum degree $2$ is a potential game.
\end{theorem}

\begin{proof}
Assume that there are no stubborn agents and there are at least 3 nodes in $G$. Consider an improving move of an agent $\calA$ in assignment $C_0$ that leads to assignment $C_1$. Suppose $\calA$ moves from node $s$ (in $C_0$) to node $t$ (in $C_1$). 
We first consider the case where $s$ and $t$ are neighbors: 
If $s$ has degree $1$ and $t$ has degree $2$, since this is an improving move for $\calA$, $t$'s other neighbor must have an agent of a different type from $\calA$, and thus $\Phi(C_1) - \Phi(C_0) =-m$. If, instead, $s$ has degree $2$ and $t$ has degree $2$, it must be that $\calA$'s other neighbor in $C_0$ is of the {\em same} type as $\calA$ or empty, while $t$ is adjacent to an agent of a different type than that of $\calA$. Thus, $\Phi(C_1) - \Phi(C_0) \leq -m$.

Next, we look at the cases where $s$ and $t$ are not neighbors:
\begin{itemize}
\item $s$ and $t$ have degree $1$. Then, $s$'s neighbor be an agent of the same type as $\calA$ or it is an empty node, while $t$'s neighbor must be an agent of type different than that of $\calA$. Hence, $\Phi(C_1) - \Phi(C_0) \leq -m$.

\item $s$ has degree 1 and $t$ has degree $2$. Then, either $\calA$ improves its utility from $0$ to $1/2$ and thus $\Phi(C_1) - \Phi(C_0) \leq 1 - 2m$, or improves its utility from $0$ to $1$ and thus $\Phi(C_1) - \Phi(C_0) \leq -m$.

\item $s$ has degree 2 and $t$ has degree $1$. Then, either $\calA$ improves its utility from $0$ to $1$ and thus $\Phi(C_1) - \Phi(C_0) \leq -m$, or improves its utility from $1/2$ to $1$ and thus $\Phi(C_1) - \Phi(C_0) = 2m-1-m=-m-1$.

\item $s$ and $t$ have degree $2$. Then, either $\calA$ improves its utility from $0$ to $1/2$ and thus $\Phi(C_1) - \Phi(C_0) \leq 1-2m$, or from $0$ to $1$ and thus $\Phi(C_1) - \Phi(C_0) \leq -m$, or from $\frac{1}{2}$ to $1$ and thus $\Phi(C_1) - \Phi(C_0)\leq m-1$.
\end{itemize}
In all possible cases,  $\Phi(C_1) - \Phi(C_0)<0$ for all values of $m \in (0,1)$. Therefore, $\Phi$ is a generalized ordinal potential function, and the \game is a potential game. From Lemma~\ref{lem:stubborn}, the result is also true when stubborn agents exist. 
\end{proof}

\section{Efficiency at Equilibrium} \label{sec:poa}
In this section, we turn our attention to showing welfare guarantees at equilibrium by bounding the price of anarchy and the price of stability. We show bounds on the price of anarchy for games consisting of $n \geq 2$ {\em strategic} agents that are partitioned into $k \geq 2$ types. The types might be {\em asymmetric} or {\em symmetric}. In the asymmetric case, each type $T$ consists of some number $n_T \geq 1$ of agents without any restriction other than that the types form a partition of the set of all agents. In contrast, in the symmetric case, each type $T$ consists of the same number $n_T = n/k$ of agents. We show that the price of anarchy is at most $n/(k-1)$ for asymmetric types and at most $k/(k-1)$ for symmetric type. We complement these upper bounds with tight lower bounds, which in some cases are even achieved in games with a topology as simple as a line. For the price of stability, we provide a lower bound of $65/62 \approx 1.048$ for $k=2$ types showing that there are games in which the optimal assignments (in terms of social welfare) are not always equilibria. 

We start be showing the upper bounds on the price of anarchy. 

\begin{theorem} \label{thm:poa-upper}
For $k \geq 2$ types, the price of anarchy of any \game
with $n$ agents is at most $\frac{n}{k-1}$ if the types are asymmetric, and 
at most $\frac{k}{k-1}$ if the types are symmetric. 
\end{theorem}

\begin{proof}
Consider an arbitrary game $I$ and denote by $C$ some equilibrium assignment, in which there is an empty node $v$ that is adjacent to $x_T$ agents of type $T$. Let $x = \sum_T x_T \geq 1$. First, suppose that $x = x_T = 1$ for some type $T$. Then the agents of all types different than $T$ must have utility $1$ to not have incentive to deviate to $v$, which means that $\SW(C) \geq \sum_{t \neq T} n_t$. Since the maximum utility of each agent is $1$, $\OPT(I) \leq n$, leading to a price of anarchy of at most $\frac{n}{n-n_T}$. 
If the types are asymmetric, then $n_T \leq n-k+1$, and the price of at anarchy is at most $\frac{n}{k-1}$. 
If the types are symmetric, then $n_T = n/k$, and the price of anarchy is at most $\frac{k}{k-1}$. 

So, now we assume that $x \geq 2$. 
To not have incentive to deviate to $v$, any agent of type $T$ that is adjacent to $v$ must have utility at least 
$\frac{\sum_{t \neq T}x_t}{x - 1} = \frac{x-x_T}{x-1}$, 
while any agent of type $T$ that is not adjacent to $v$ must have utility at least 
$\frac{\sum_{t \neq T}x_t}{x} = \frac{x-x_T}{x}$, 
So, the social welfare at equilibrium can be lower-bounded as follows:
\begin{align*}
    \SW(C) 
    &\geq \sum_T \bigg( x_T \cdot \frac{x-x_T}{x - 1}  + (n_T - x_T)\cdot \frac{x-x_T}{x} \bigg) \\
    &\geq \frac{1}{x} \sum_T n_T \cdot (x-x_T) \\
    &= n - \frac{1}{x}\sum_T n_T x_T.
\end{align*}
If the types are asymmetric, then $n_T \leq n-k+1$ for every type $T$, and the social welfare at equilibrium is
$$\SW(C) \geq n - \frac{1}{x}\sum_T (n-k+1) x_T = n - (n-k+1) = k-1.$$
If the types are symmetric, then $n_T = n/k$ for every type $T$, and the social welfare at equilibrium is
$$\SW(C) \geq n - \frac{1}{x}\sum_T \frac{n}{k} x_T = n - \frac{n}{k} = n \frac{k-1}{k}.$$
Since $\OPT(I) \leq n$, we conclude that the price of anarchy is at most $\frac{n}{k-1}$ if the types are asymmetric, and at most $\frac{k}{k-1}$ if the types are symmetric.
\end{proof}

Next, we investigate under which conditions the upper bounds of Theorem~\ref{thm:poa-upper} are tight. We start with the case of asymmetric types. We show that for games in which the topology is a line, a better price of anarchy bound can be shown; in particular, the price of anarchy can be shown to be at most a function of $n$ that tends to $2$, for any number of types and any distribution of agents into types. 

\begin{theorem} \label{thm:poa-asymmetric-line}
The price of anarchy of any \game with a line topology and $n$ agents of $k \geq 2$ types is at most $\frac{2n}{n-1}$. 
\end{theorem}

\begin{proof}
Consider an arbitrary game $I$ with a line topology, and let $C$ be some equilibrium assignment. First, suppose that there is an empty node $v$ in $C$ that is adjacent to only agents of one type, say red. Then, to not have incentive to deviate to $v$, the non-red agents must all have utility $1$. 
Observe that there must exist some red agent with utility at least $1/2$; otherwise there would exist some empty node  that is adjacent to non-red agents (since $G$ is a connected line), where the red agents could jump to get positive utility. So, $\SW(C) \geq n - n_r + 1/2$, where $n_r$ is the number of red agents. We consider the following two cases:
\begin{itemize}
    \item If $n_r \leq \lfloor n/2 \rfloor +1$, then since $\OPT(I) \leq n$ 
    and $\SW(C) \geq n-n_r+1/2 \geq \lceil n/2 \rceil -1/2$, the price of anarchy is at most $\frac{n}{\lceil n/2 \rceil-1/2}$. This is at most $\frac{2n}{n-1}$  if $n$ is even, and at most $2$ if $n$ is odd. 
    
    \item If $n_r > \lfloor n/2 \rfloor + 1$, then it is not possible to achieve social welfare $n$ since some red agents cannot have only neighbors of different type. In particular, the maximum possible social welfare is $\OPT(I) \leq 2(n-n_r)+1 = 2(n-n_r+1/2)$; this can be achieved by creating $n-n_r$ pairs of agents of different type (which, in the optimal assignment, are placed one after the other on the line) and possibly one more red agent placed at the beginning or the end of the line getting utility 1. Hence, the price of anarchy is at most $2$.
\end{itemize}

Finally, suppose that the empty node $v$ in $C$ is adjacent to agents of two different types, say red and blue. Then, to not have incentive to deviate to $v$, the non-red and non-blue agents must all have utility $1$, the two agents adjacent to $v$ must also have utility $1$, and the remaining red and blue agents must have utility at least $1/2$. Therefore, since every agent gets utility at least $1/2$ and the maximum utility is $1$, the price of anarchy is at most $2$. 
\end{proof}

We next show that the price of anarchy approaches the upper bound for asymmetric types when the topology is slightly more general than a line. In particular, we show the following lower bound. 

\begin{theorem}
There is a \game with a star topology and and $n$ agents of $k \geq 2$ asymmetric types such that the price of anarchy at least $\frac{n-1}{k-1}$.
\end{theorem}

\begin{proof}
Consider a game with star topology and $k \geq 2$ types such that $k-1$ types have size $1$ while one type, say red, has size $n-k+1$. Any assignment such that the center node of the topology is occupied by an agent of the first $k-1$ types is optimal with social welfare $n$, whereas the assignment according to which the center node is occupied by a red agent is an equilibrium with social welfare $k-1 + \frac{k-1}{n-1} = (k-1) \frac{n}{n-1}$. Hence, the price of anarchy is at least $\frac{n-1}{k-1}$.
\end{proof}

We now turn our attention to the case of games with symmetric types. We first show that for $k=2$ symmetric types, a slightly better bound than that of Theorem~\ref{thm:poa-asymmetric-line} can be shown when the topology is a line; this bound is a different function of $n$ that again tends to the upper bound of $2$ when $n$ becomes large. 

\begin{theorem}
For \games with line topology and $n$ agents of $k=2$ symmetric types, the price of anarchy is exactly $\frac{2n}{n+4}.$
\end{theorem}

\begin{proof}
For the lower bound, consider a game with a line topology consisting of $n+1$ nodes and $n$ agents of two types, red ($r$) and blue ($b$).
In the optimal assignment, the agents occupy the nodes of the line so that each red agent is followed by a blue agent. 
This guarantees that every agent has only neighbors of different type and thus achieves a maximum utility of $1$, leading to an optimal social welfare of $n$. 
Now consider the assignment in which the agents are ordered according to their types as follows: $(r,b,b,r,r,\ldots,r,r,b,b,r,v,b,r)$, where $v$ is the empty node. This is an equilibrium since the empty node is adjacent to a red and a blue agent, and every agent has either only neighbors of different type (and thus utility $1$), or one red and one blue neighbor (and thus utility $1/2$); hence, no agent would have incentive to jump to the empty node since it only offers utility $1/2$. There are exactly $4$ agents that achieve utility $1$ and $n-4$ agent that achieve utility $1/2$, for a social welfare of $(n-4)/2 +4 = (n+4)/2$. Consequently, the price of anarchy is at least $2n/(n+4)$.

For the upper bound, let $C$ be an equilibrium assignment. First observe that, if there is an empty node $v$ in $C$ that is adjacent to a single agent or two agents of the same type, say red, then it has to be the case that all blue agents have utility $1$, and consequently the same holds for the red agents due to symmetry, leading to price of anarchy of $1$. So, we can assume that $v$  is adjacent to one red agent $\calA$ and one blue agent $\calB$. Since $C$ is an equilibrium, both $\calA$ and $\calB$ must have a neighbor of different type and utility $1$ since otherwise they would jump to $v$ to connect to each other. In addition, since the topology is a line, there must be at least two more agents with a single neighbor and utility $1$; otherwise they would have utility $0$ and incentive to jump to $v$ for a utility of $1/2$. Similarly, all remaining agents must already achieve utility at least $1/2$ to not have incentive to jump to $v$. Hence, the social welfare at any equilibrium assignment must be at least $(n-4)/2 + 4 = (n+4)/2$. Since the maximum possible social welfare is $n$, the price of anarchy is at most $2n/(n+4)$. 
\end{proof}

For games with $k \geq 3$ symmetric types and line topology, we show a bound of $\frac{n}{(k-1)\frac{n}{k}+1/2}$ which almost matches the upper bound of $k/(k-1)$.  

\begin{theorem}
For  \games with a line topology and $n$ agents of  $k \geq 3$ symmetric types such that the price of anarchy is exactly $\frac{n}{(k-1)\frac{n}{k}+\frac12}$.
\end{theorem}

\begin{proof}
For the lower bound, consider a game with a line topology consisting of $n+1$ nodes. Suppose that the $k$ types are identified by the integers $\{1, \ldots, k\}$.  An optimal assignment is to arrange the agents so that an agent of type $T$ is followed by an agent of type $T+1$, for every $T \leq [k-1]$. Hence, every agent gets utility $1$ for an optimal social welfare of $n$. Now consider the following assignment: There is a repeated sequence of agents so that an agent of type $T$ is followed by an agent of type $T+1$, for every $T \in \{1, \ldots, k-2\}$. This sequence is followed by all agents of type $k$ and the last node is left empty. Clearly, all agents of the first $k-1$ types have utility $1$ and no incentive to deviate to the empty node, while there is exactly one agent of type $k$ with utility $1/2$ and all others have utility $0$. So, this is an equilibrium with social welfare $(k-1)\frac{n}{k} + 1/2$, leading to the desired lower bound on the price of anarchy. 

For the upper bound, let $C$ be an equilibrium assignment. We consider the following two cases:
\begin{itemize}
    \item There is an empty node $v$ in $C$ that is adjacent to a single agent or two agents of the same type, say red. 
    Observe that all non-red agents must have utility $1$ so that they do not have incentive to deviate to $v$. In addition, at least one of the red agents must be adjacent to an agent of different type and thus have utility at least $1/2$; otherwise, if all red agents are adjacent only to other red agents or empty nodes, then there must exist a non-red agent that is adjacent to an empty node where the red agents would like to jump. Hence, the social welfare of this equilibrium is at least $(k-1)\frac{n}{k} + 1/2$. Since the optimal social welfare is at most $n$, we get the desired upper bound on the price of anarchy.

    \item There is an empty node $v$ in $C$ that is adjacent to two agents of different type, say red and blue. 
    Then all non-red and non-blue agents must have utility $1$ so that they do not have incentive to deviate to $v$. In addition, both agents that are adjacent to $v$ must have utility $1$ to not have incentive to deviate to $v$ and connect to each other, while all other red and blue agents must have utility at least $1/2$. Hence, the social welfare of this equilibrium is at least $(k-2)\frac{n}{k} + 2 + 2\left(\frac{n}{k}-1 \right)\cdot \frac12 = (k-1)\frac{n}{k} + 1 \geq (k-1)\frac{n}{k} + 1/2$, leading to the desired upper bound again. 
\end{itemize}
This completes the proof. 
\end{proof}

We conclude this section with a lower bound of $65/62$ on the price of stability, showing that the assignment with optimal social welfare is not always an equilibrium. 

\begin{theorem}
There is a \game with $k=2$ types such that the price of stability is at least $65/62 \approx 1.048$.
\end{theorem}

\begin{proof}
Consider a \game with $2$ red agents and $4$ blue agents. The topology is as shown in Figure~\ref{fig:pos-lower}(a). 
The assignment of the agents in Figure~\ref{fig:pos-lower}(b) is the optimal one with social welfare $65/12$, while that in Figure~\ref{fig:pos-lower}(c) is the best equilibrium with social welfare $62/12 \approx 5.16$. We will now argue that any equilibrium assignment $C$ has social welfare at most $62/12$, yielding the desired lower bound on the price of stability. For simplicity in our notation, for any node $v$, we will denote by $\calA(v)$ the agent that occupies it (if any). 

\begin{figure}[t]
    \centering
    \includegraphics[scale=0.5]{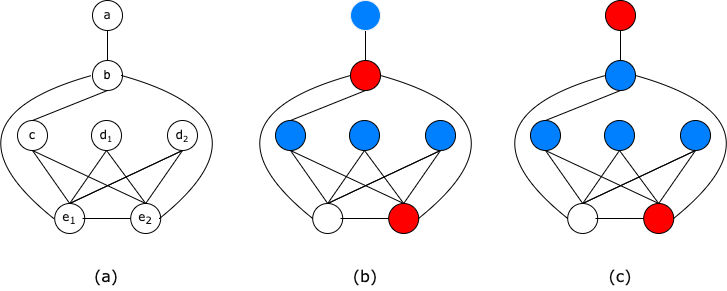}
    \caption{(a) The topology of the game used to show the lower bound on the price of stability; (b) The optimal assignment with social welfare $65/12$; (c) The best equilibrium assignment with social welfare $62/12$.}
    \label{fig:pos-lower}
\end{figure}

First note that if there is an agent with utility $0$, or two agents with utility at most $1/2$, then the social welfare is bounded above by $5$, and we are done. So, we can assume that $C$ has the following two properties: 
(1) No agent has utility 0, and (2) at most one agent has utility $\leq 1/2$. 
Property (1) implies immediately a third property: (3) Either node $a$ is empty, or $a$ and $b$ are occupied by agents of different types. 
This follows since agent $\calA(a)$ would have utility $0$ if node $b$ is empty, or if both $a$ and $b$ are occupied by agents of the same type. 

Next, observe that if the degree-2 nodes $d_1$ and $d_2$ are occupied by agents of different types, then either they both have utility $1/2$ (violating Property (2)) or one of them has utility $0$ (violating Property (1)). Therefore it must be the case that either one of $d_1$ and $d_2$ is left empty, or both are occupied by agents of the same type. We consider these two cases separately: 

\medskip
\noindent 
{\bf Case 1: $d_2$ is empty; the case where $d_1$ is empty is symmetric.} 
If $\calA(d_1)$ is red, then since it must obtain utility strictly larger than $0$ and one of the agents at $a$ or $b$ must be red by Property (3), it must be the case that $c$, $e_1$ and $e_2$ are all occupied by blue agents and have utility at least $1/2$, a contradiction to Property (2). 
If on the other hand $\calA(d_1)$ is blue, then since either $a$ or $b$ must be occupied by a red agent by Property (3), $e_1$ and $e_2$ must be occupied by agents of different type,say $e_1$ is occupied by a red agent and $e_2$ is occupied by a blue agent. 
However, this means that $\calA(d_1)$ and $\calA(e_2)$ have utility at most $1/2$, a contradiction to Property (2).

\medskip
\noindent 
{\bf Case 2: $d_1$ and $d_2$ are occupied by agents of the same type.}
If the agents $\calA(d_1)$ and $\calA(d_2)$ are red, then all the blue agents must be at the remaining nodes, which means that at least one of them must get utility $0$, thus contradicting Property (1). So, the agents $\calA(d_1)$ and $\calA(d_2)$ must both be blue. 
To satisfy Properties (1) and (2), either both $e_1$ and $e_2$ are both occupied by red agents, or one of them is empty and the other is occupied by a red agent. 
\begin{itemize}
    \item {\bf Both $\calA(e_1)$ and $\calA(e_2)$ are red}. 
    To satisfy Property (3), since it is not possible to assign different-type agents at $a$ and $b$, it has to be the case that $a$ is empty. However, such an assignment cannot be an equilibrium as the red agents at $e_1$ and $e_2$ have utility less than $1$ (since they are connected to each other) and would jump to $a$ to get utility $1$. 

    \item {\bf One of $e_1$ and $e_2$ is occupied by a red agent, while the other is empty.} 
    Suppose without loss of generality that $e_1$ is empty and $e_2$ is occupied by a red agent. 
    Since $a$ and $b$ must have agents of different types due to Property (3), only two assignments are possible. 
    The first one is when $\calA(a)$ is blue and $\calA(b)$ is red. This is the optimal assignment with social welfare $65/12$, which however is not an equilibrium since $\calA(b)$ has utility $2/3$ and would prefer to deviate to the empty node $e_1$ to get utility $3/4$; see Figure~\ref{fig:pos-lower}(b). 
    The second one is when $\calA(a)$ is red and $\calA(b)$ is blue. This is the best possible equilibrium with social welfare $62/12$; see Figure~\ref{fig:pos-lower}(c).
\end{itemize}
The proof is now complete. 
\end{proof}

\section{Open Problems}
In this paper we proposed a new class of \games in networks. We showed that determining if there exists an equilibrium assignment in a given instance is NP-hard, using a reduction that relies on the existence of stubborn agents. Do games with only strategic agents always admit an equilibrium? If not, what is the complexity of deciding if an equilibrium exists in a given \game with only strategic agents? 

We also showed that the game is not potential for tree topologies; however, if all agents are strategic, there is always an equilibrium, regardless of the number of empty nodes in the tree. It would be interesting to know if the game is weakly acyclic in trees. For spider graphs with a single empty node, we showed that the game is potential. We conjecture the same holds for any number of empty nodes. While games on $\delta$-regular graphs are always potential games when there is a single empty node, there are games that are not potential when there are $4$ or more empty nodes. It would be very interesting to know whether the game is potential for up to $3$ empty nodes. Similar questions would be interesting for other classes of graphs, such as bipartite graphs and grids. 

Finally, while our price of anarchy bounds are essentially tight, it would be interesting to further investigate the price of stability and either improve the lower bound of $65/62$ that we managed to show for two types, or show an upper bound for it that is better than the upper bounds shown for the price of anarchy. 

\section*{Acknowledgments}
The authors gratefully acknowledge many useful discussions with J. Opatrny, as well as several useful and constructive comments of anonymous referees.

\bibliographystyle{plainnat}
\bibliography{main}

\end{document}